\documentclass[acmsmall,screen]{acmart}

\usepackage{comment}

\AtBeginDocument{%
  }

\usepackage[font=small,labelfont=bf]{caption}
\usepackage{subcaption}
\usepackage[clock]{ifsym}

\usepackage{comment}
\usepackage{graphicx}
\usepackage{algorithmic}
\usepackage{algorithm}

\usepackage{amsmath,bm,esint}
\usepackage{blkarray, bigstrut}
\usepackage{xparse}
\usepackage{rotating}
\usepackage{multirow}
\usepackage{tabularx}
\usepackage{tablefootnote}
\usepackage{booktabs}
\usepackage{array}
\usepackage{colortbl}

\usepackage{lstautogobble}

\definecolor{light-gray}{gray}{0.95}
\definecolor{tbl-row-color}{gray}{0.95}
\usepackage{listings}

\definecolor{keywordsColor}{RGB}{134, 14, 11}
\definecolor{commentsColor}{RGB}{54, 54, 54}

\lstdefinelanguage{OurLanguage}{
    alsoletter={:,=, <, >, &, |},
    keywords={while, end, types, if, else, else:, and, or, not, Normal},
    backgroundcolor=\color{light-gray},
    morekeywords={=, <, >, <=, >=, ==, !=, &&, ||},
    basicstyle={\ttfamily\small\normalfont},
    keywordstyle={\color{keywordsColor}\ttfamily\bfseries},
    comment=[l]{\#},
    commentstyle={\color{commentsColor}\ttfamily},
    autogobble=true,
    mathescape=true
}
\lstdefinestyle{program}{basicstyle=\small\ttfamily,keywordstyle=\bfseries}

\usepackage{tikz}
\usetikzlibrary{shapes.geometric, arrows}
\tikzstyle{startstop} = [rectangle, rounded corners, minimum width=3cm, minimum height=1cm,text centered, draw=black, fill=red!30]

\tikzstyle{io} = [rectangle, minimum width=4.5cm, minimum height=1.8cm, text centered,  text width=4cm, draw=black, fill=white!30]

\tikzstyle{process} = [rectangle, minimum width=3cm, text width=8cm, minimum height=1cm, text centered, draw=black, fill=orange!30]

\tikzstyle{decision} = [diamond, minimum width=3cm, minimum height=1cm, text centered, draw=black, fill=green!30]

\tikzstyle{arrow} = [thick,->,>=stealth, text width = 110]

\newcommand{\E}{\mathbb{E}}
\newcommand{\var}{\mathbb{V}\mathrm{ar}} 

\newcommand{\pr}{\mathbb{P}}

\newcommand{\real}{{\mathbb R}}
\newcommand{\nat}{{\mathbb N}}

\newcommand{\X}{{\mathbf X}}

\newcommand{\Y}{{\mathbf Y}}
\newcommand{\D}{{\mathbf D}}

\newcommand{\Z}{{\mathbf Z}}

\newcommand{\F}{{\mathbf F}}

\newcommand{\Ncal}{\mathcal{N}}
\newcommand{\Pcal}{\mathcal{P}}

\newtheorem*{remark}{Remark}

\definecolor{BrickRed}{HTML}{B6321C}
\definecolor{BlueViolet}{HTML}{473992}
\definecolor{Maroon}{HTML}{AF3235}
\definecolor{ForestGreen}{HTML}{009B55}

\begin{document}

\title{Exact and Approximate Moment Derivation for Probabilistic Loops With Non-Polynomial Assignments}
\author{Andrey Kofnov}
\email{andrey.kofnov@tuwien.ac.at}
\orcid{0000-0002-1734-2918}
\affiliation{%
  \institution{Faculty of Mathematics and Geoinformation, TU Wien}
  \streetaddress{Wiedner Hauptstrasse 8-10}
  \city{Vienna}
  \state{Vienna}
  \country{Austria}
  \postcode{1040}
}

\author{Marcel Moosbrugger}
\affiliation{%
  \institution{Faculty of Informatics, TU Wien}
  \streetaddress{Favoritenstrasse 11}
  \city{Vienna}
  \postcode{1040}
  \country{Austria}}
\email{marcel.moosbrugger@tuwien.ac.at}
\orcid{0000-0002-2006-3741}

\author{Miroslav Stankovi{\v{c}}}
\affiliation{%
  \institution{Faculty of Informatics, TU Wien}
  \city{Vienna}
  \postcode{1040}
  \country{Austria}}
\orcid{0000-0001-5978-7475}

\author{Ezio Bartocci}
\affiliation{%
 \institution{Faculty of Informatics, TU Wien}
 \streetaddress{}
 \city{Vienna}
  \postcode{1040}
  \country{Austria}}
\orcid{0000-0002-8004-6601}

\author{Efstathia Bura}
\affiliation{%
  \institution{Faculty of Mathematics and Geoinformation, TU Wien}
  \streetaddress{Wiedner Hauptstrasse 8-10}
  \city{Vienna}
  \state{Vienna}
  \postcode{1040}
  \country{Austria}}
\email{efstathia.bura@tuwien.ac.at}
\orcid{0000-0003-4972-5320}

\renewcommand{\shortauthors}{Kofnov et al.}

\begin{abstract}
Many stochastic continuous-state dynamical systems can be modeled as probabilistic programs with nonlinear non-polynomial updates in non-nested loops. We present two methods, one approximate and one exact, to automatically compute, without sampling, moment-based invariants for such probabilistic programs as closed-form solutions parameterized by the loop iteration. The exact method applies to probabilistic programs with trigonometric and exponential updates and is embedded in the \textsc{Polar} tool. The approximate method for moment computation applies to any nonlinear random function as it exploits the theory of polynomial chaos expansion to approximate non-polynomial updates as the sum of orthogonal polynomials. This translates the dynamical system to a non-nested loop with polynomial updates, and thus renders it conformable with the \textsc{Polar} tool that computes the moments of any order of the state variables. We evaluate our methods on an extensive number of examples ranging from modeling monetary policy to several physical motion systems in uncertain environments. The experimental results demonstrate the advantages of our approach with respect to the current state-of-the-art.
\end{abstract}

\ccsdesc[300]{Theory of computation~Random walks and Markov chains}
\ccsdesc[500]{Theory of computation~Probabilistic computation}
\ccsdesc[500]{Theory of computation~Program analysis}

\keywords{Probabilistic programs, Prob-solvable loops, Polynomial Chaos Expansion, Non-linear updates, Trigonometric updates, Exponential updates, Stochastic dynamical systems}


\maketitle

\section{Introduction}

Probabilistic programs (PPs) are modern tools to automate statistical modeling. They are becoming ubiquitous in AI applications, security/privacy protocols, and stochastic dynamical system modeling. PPs translate stochastic systems into programs whose execution gives rise to sets of random variables of unknown distributions. In the case of dynamical stochastic systems, corresponding PPs incorporate dynamics via loops, in which case, the distributions of the generated random quantities also vary along loop iterations. 
Automating statistical inference for these stochastic systems requires knowledge of their distribution; that is, the distribution(s) of the random variable(s) generated by executing the probabilistic program that encodes them.

Statistical moments are essential quantitative measures that characterize many probability distributions. In~\cite{Bartoccietal2019} the authors introduced the notion of \textit{Prob-solvable loops}, a class of probabilistic programs with a non-nested loop with polynomial updates and acyclic state variable dependencies for which it is possible to automatically compute moment-based invariants of any order over the program state variables as closed-form expressions in the loop iteration. This approach was first implemented in the \textsc{Mora}~\cite{BartocciKS20a} tool and later further improved in the \textsc{Polar} tool~\cite{Moosbruggeretal2022} to also support multi-path probabilistic loops with if-statements, symbolic constants, circular linear dependency among program state variables and drawing from distributions that depend on program state variables.  More recently, \cite{AmrollahiBKKMS22} proposed a method to handle more complex state variable dependencies that make both probabilistic and deterministic loops in general \emph{unsolvable}. Furthermore, the work in~\cite{KarimiMSKBB22} shows how to use a finite set of high-order moment-based invariants to estimate the probability distribution of the program's random variables. The core theory underlying all these approaches combines techniques from computer algebra such as symbolic summation and recurrence equations~\cite{KauersP11} with statistical methods.

Despite the successful application of these methods and tools in many different areas, including the analysis of consensus/security protocols~\cite{Moosbruggeretal2022}, inference problems in Bayesian networks~\cite{StankovicBK22,BartocciKS20a} and automated probabilistic program termination analysis~\cite{MoosbruggerBKK21,MoosbruggerBKK21b}, they fail when modeling more complex dynamics that require non-polynomial updates. Such examples are depicted in Fig.~\ref{fig:taylor_model}, 
where the updates of the variables contain the logarithmic function, and in Fig.~\ref{fig:turningexample}, where modeling the physical motion of a vehicle requires trigonometric functions.
Thus, how to leverage the class of \textit{Prob-solvable loops} to compute moment-based invariants as closed-form expressions in probabilistic loops with non-polynomial updates remains an open research problem.

In preliminary work presented at QEST 2022~\cite{KofnovMSBB22}, we provided a solution to this problem leveraging the theory of \emph{general Polynomial Chaos Expansion} (gPCE) \cite{XiuKarniadakis2002a}, which consists of decomposing a non-polynomial random function into a linear combination of orthogonal polynomials. gPCE theory, upon which our approach is based, assures that the polynomial approximation of non-polynomial square-integrable functions converges to the truth by increasing the degree of the polynomial and guarantees the estimation of moments of random variables with complex probability distributions. Once such a polynomial approximation is applied, we take advantage of the work in \cite{Bartoccietal2019,BartocciKS20a} to automatically estimate the moment-based invariants of the loop state variables as closed-form solutions. 
 In Fig.~\ref{fig:taylor_model} we illustrate our gPCE-based approach via the Taylor rule in monetary policy, where we estimate the expected interest rate  given a target inflation rate and the gross domestic product (GDP). In this example, we approximate the original log function with $5$th degree polynomials and obtain a Prob-solvable loop.  This enables the automatic computation of the gPCE approximation of the moments in closed-form  at each loop iteration ($n$) using the approach proposed in~\cite{Bartoccietal2019}. 
 
 Trigonometric functions are prevalent in stochastic dynamical systems of motion. The exponential function is directly related to trigonometric functions as well as to characteristic functions of distributions. We combine the methodology proposed in ~\cite{Jasouretal2021} with Prob-Solvable loops to obtain exact moments of trigonometric and exponential functions of random variables at loop iteration.
 Fig.~\ref{fig:turningexample} presents the PP encoding of the stochastic dynamical model of a turning vehicle that requires trigonometric updates. Our new approach incorporates results in~\cite{Jasouretal2021} and computes moments of all orders in closed form as a function of iteration number.
 In the QEST 2022 paper~\cite{KofnovMSBB22}, we provided only gPCE-based moment estimates of trigonometric and exponential updates. Fig.~\ref{fig:turningexample} shows how our novel approach provides the exact expected trajectory $(x_n,y_n)$ as a function of the loop iteration $n$. Although the expected trajectory moves from left to right, we can see that some of the sampled trajectories in effect turn backward.
 This underlines the necessity for high-order statistical moments, which our approaches are able to compute.
 
 \paragraph{Paper contribution.} This paper  extends and improves our previous QEST 2022 conference~\cite{KofnovMSBB22} manuscript with the following new contributions:
\begin{itemize}
    \item[(i)] \cite{Jasouretal2021} developed a method that obtains the exact time evolution of the moments of random states for a class of dynamical systems that depend on trigonometric updates. We amended their approach and make it compatible with the \texttt{Polar} tool \cite{Moosbruggeretal2022}. Specifically, we  incorporated the approach of~\cite{Jasouretal2021} into \emph{Prob-solvable loops} when updates involve trigonometric functions. This allows us to automatically compute the \emph{exact} moments of any order and at all iterations. Moreover, we extended \cite{Jasouretal2021} to include \emph{exponential} updates. We present the new methodological material in Sec.~\ref{sec:exact}.
    \item[(ii)] We rewrote the abstract and the introduction to reframe our work with the new material.
    We updated \emph{Related Work} to reflect our new contributions and compare them with the current state-of-the-art. We revised the text of all the other sections adding new examples.
    \item[(iii)] 
    We have considerably improved and expanded the evaluation section by adding six benchmark models to the previous five in \cite{KofnovMSBB22}.
    Moreover, we extend our original evaluation by including comparisons for our newly proposed exact method.
\end{itemize}

\begin{figure}[!t]
\centering
  \includegraphics[scale=.40]{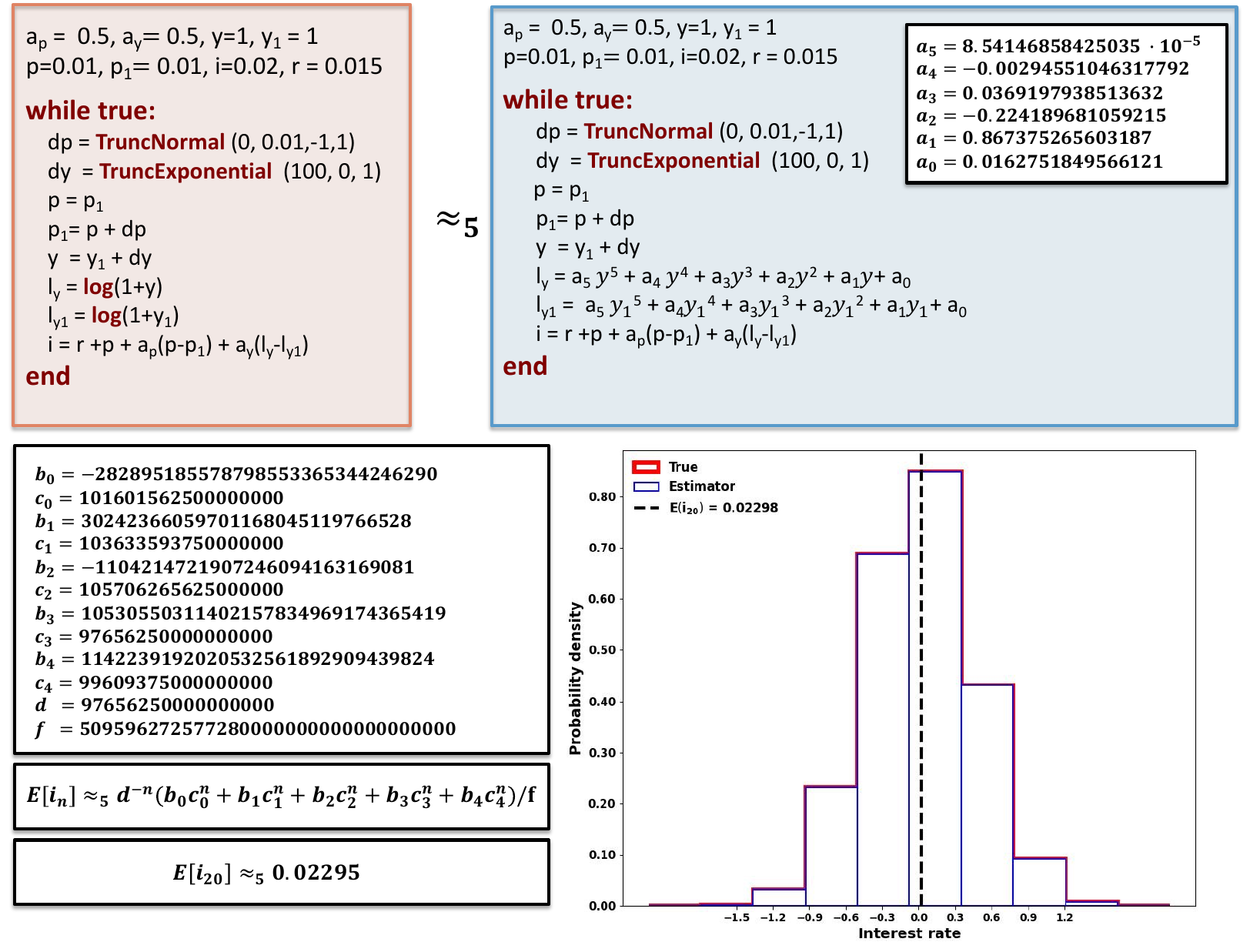}
  \caption{The probabilistic loop in the top left panel encodes the Taylor rule~\cite{Taylor1993}, an equation that prescribes a value for the short-term interest rate based on a target inflation rate and the gross domestic product.
  The program uses a non-polynomial function (log) in the loop body to update  the continuous-state variable ($i$). The top right panel contains the \emph{Prob-Solvable loop} (with polynomial updates) obtained by approximating the log function using polynomial chaos expansion (up to $5$th degree).  In the bottom left, we compute the expected interest rate ($\E[i_n]$) as a closed-form expression in loop iteration $n$ using the Prob-solvable loop and evaluate it at $n=20$.  In the bottom right panel, we compare the true and  estimated distributions for a fixed iteration (we sample the loop  $10^6$ times at iteration $n=20$).} 
  \label{fig:taylor_model}
\end{figure}

\begin{figure}[!t]
\centering
  \includegraphics[scale=.40]{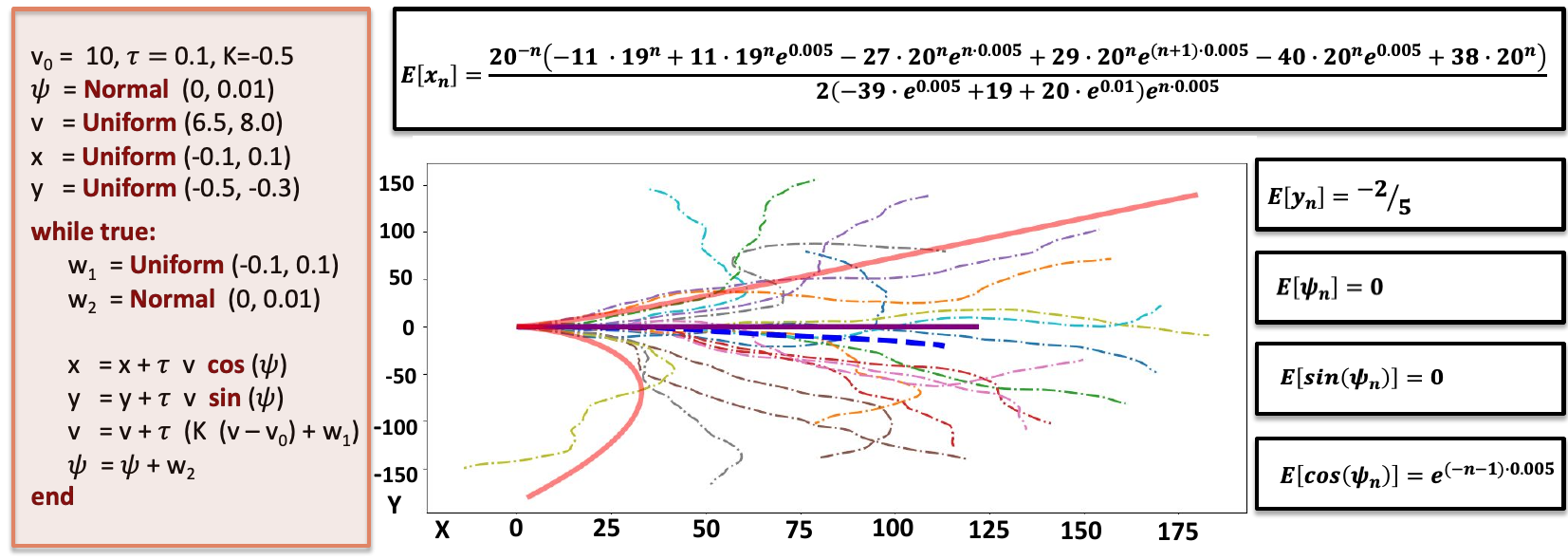}
  \caption{On the left is a probabilistic loop modeling the behavior of a turning vehicle~\cite{Srirametal2020} using non-polynomial (cos, sin) updates in the loop body. On the top right is the exact expected position $(x_n,y_n)$ and other exact expected values
  computed automatically in closed-form in the number of loop iterations $n$.
  The plot in the center contains $20$ sampled trajectories $(x_n,y_n)$ up to iteration $n=201$ (dashdot lines with different colors) and the approximated expected trajectory computed by averaging the sampled ones (dashed blue line).
  Moreover, we automatically computed the exact expected trajectory and standard deviation with our method. The solid purple line marks the exact expected trajectory.
  The two solid red lines mark the boundary of the region contained within $\pm$ two standard deviations of the expected trajectory.}
  \label{fig:turningexample}
\end{figure}

\paragraph{Related Work.} \emph{Taylor series models}~\cite{Revol2005,Neher2007,Chenetal12} are a well-established computational tool in reachability analysis for (non-probabilistic) non-linear dynamical systems that combine the polynomial approximation of Taylor series expansion
with the error intervals to over-approximate the set of dynamical trajectories for a finite time horizon.
\cite{Srirametal2020} approximates non-polynomial functions of random variables as polynomials using Taylor series expansion. Other works~\cite{Stankovic1996,Triebel2001} follow a similar approach.  Such approximations work well in the neighborhood of a point, requiring otherwise a polynomial with a high degree to maintain a good level of accuracy. This approach is not suitable for approximating functions with unbounded support.

Our two approaches, exact and gPCE based, have several advantages over the method proposed in~\cite{Srirametal2020}.  First, in contrast, to~\cite{Srirametal2020}, neither of our methods is  limited to a fixed iteration. Instead, we compute closed-form expressions in the number of loop iterations. Second, the interval estimates in ~\cite{Srirametal2020} become larger after a few iterations. Our exact moment calculation for the same models, involving trigonometric functions, is not affected by iteration number and is exact (i.e., incurs no error) at all iterations.

The method of \cite{Jasouretal2021}, which we adjusted, extended, and made compatible with \texttt{Polar}, concerns discrete-time stochastic nonlinear dynamical systems subjected to probabilistic uncertainties. \cite{Jasouretal2021} focused on nonlinear autonomous and robotic systems where motion dynamics are described in terms of translational and rotational motions over the planning horizon. The latter naturally led to the introduction of trigonometric and mixed-trigonometric-polynomial moments in order to obtain an exact description of the moments of uncertain states. This method computes exact moments but it can only handle systems encoded in PPs, where all nonlinear transformations take standard, trigonometric, or mixed-trigonometric polynomial forms.
Our approximate gPCE-based approach instead is applicable to general non-polynomial updates (e.g. containing logarithms).
Furthermore, the work of \cite{Jasouretal2021} only considers the computation of moments up to a fixed horizon.
In contrast, for systems that can be modeled as Prob-solvable loops, both our methods provide closed-form expressions parameterized by the number of loop iterations.

Polynomial chaos expansion based methods have been extensively used for uncertainty quantification in different areas, such as engineering problems of solid and fluid mechanics (e.g. \cite{GhanemSpanos1991,Fooetal2007,Houetal2006}), computational fluid dynamics (e.g., \cite{Knio2006}), flow through porous media \cite{GhanemDham1998,Ghanem1998}, thermal problems
 \cite{HienKleiber1997}, analysis of turbulent velocity fields \cite{Chorin1974,MeechamJeng1968}, differential equations (e.g., \cite{WanKarniadakis2005,XiuKarniadakis2002a}), and, more recently, geosciences and meteorology (e.g., \cite{Formaggiaetal2013,Giraldietal2017,Denamieletal2020}).

\paragraph{Outline.} Sec.~\ref{sec:preliminaries} provides the necessary background on \emph{Prob-solvable Loops} and the theory of \emph{general Polynomial Chaos Expansion} (gPCE).  Sec.~\ref{sec:PCE_algo} introduces our gPCE-based approximation method presenting the conditions that are necessary to accurately approximate general non-polynomial updates in a probabilistic loop. Sec.~\ref{sec:genfun} shows how to obtain a Prob-solvable loop using our approximation method and hence how to automatically compute moment-based invariants of all orders for the program state variables.  Sec.~\ref{sec:exact} presents the exact method leveraging the theory in~\cite{Jasouretal2021} to compute the exact moments of PPs with trigonometric and exponential updates. Sec.~\ref{sec:evaluation} evaluates the accuracy and feasibility of the proposed approaches over several benchmarks comparing them with the state-of-the-art. We conclude in Sec.~\ref{sec:conclusion}.

\section{Preliminaries}\label{sec:preliminaries}

We assume the reader to be familiar with basic probability theory.
For more details, we refer to \cite{Durrett2019}.

\subsection{Prob-Solvable Loops}\label{sec:prob-solvable-loops}

\cite{Bartoccietal2019} defined the class of \emph{Prob-solvable loops} for which  moments of all orders of program variables can be computed symbolically: given a Prob-solvable loop and a program variable $x$, their method computes a closed-form solution for $\E(x_n^k)$ for arbitrary $k \in \nat$, where $n$ denotes the $n$th loop iteration.
Prob-solvable loops are restricted to polynomial variable updates.

\begin{definition}[Prob-solvable loops \cite{Bartoccietal2019}]\label{def:probsolvable}
Let $m \in \nat$ and $x_1,\ldots x_m$ denote real-valued program
variables. 
A Prob-solvable loop with program variables $x_1,\ldots x_m$ is a loop of the form
\begin{equation*}\label{eq:ProbModel}
  I; \texttt{ while true: } U \texttt{ end}, \quad \text{such that}
\end{equation*}
\begin{itemize}
    \item $I$ is a sequence of initial assignments over a subset of $\{x_1,\ldots, x_m\}$. The initial values of $x_i$ can be drawn from a known distribution. 
    They can also be real constants.
    \item $U$ is the loop body and a sequence of $m$ random updates, each of the form,
    \begin{equation*}\label{eq:ProbModel:prob_assignments}
        x_i = \textit{Dist} \quad \text{or} \quad x_i = a_i x_i + P_{i}(x_1,\dots x_{i-1}),
    \end{equation*}
    where $a_i \in \real$, $P_{i} \in \real[x_1,\ldots,x_{i-1}]$ is a polynomial over program
    variables $x_1,\ldots,x_{i-1}$ and \textit{Dist} is a random variable whose distribution is independent of program variables with computable moments.
    $a_i$ and the coefficients in $P_i$ can be random variables with the same constraints as for \emph{Dist}.
\end{itemize}
\end{definition}

The syntax of Prob-solvable loops as defined in Definition~\ref{def:probsolvable} is restrictive.
For instance, an assignment for a variable $x_i$ must not reference variables $x_j$ with $j > i$.
Hence, the structural dependencies among program variables are acyclic.
Some of these syntactical restrictions were lifted in a later work \cite{Moosbruggeretal2022} to support distributions depending on program variables, if-statements, and linear cyclic dependencies.
The latter means that polynomial assignments can be of the form $x_i = L_{i}(x_1, \dots, x_n) + P_{i}(x_1,\dots x_{i-1})$, where $P_{i}$ is a polynomial, and $L_{i}$ is a linear function, as long as all program variables in $L_{i}(x_1, \dots, x_n)$ with non-zero coefficient depend only linearly on $x_i$.
In this work, we utilize this relaxation and allow for linear cyclic dependencies in Prob-solvable loops.

Many real-life systems exhibit non-polynomial dynamics and require more general updates, such as, for example, trigonometric or exponential functions.
In this work, we develop two methods -- one approximate, one exact -- that allow the modeling of non-polynomial assignments in probabilistic loops by polynomial assignments.
Doing so allows us to use the \emph{Prob-solvable loop} based methods in \cite{Bartoccietal2019,Moosbruggeretal2022} to compute the moments of the stochastic components of a much broader class of systems.
Our method for exact moment derivation for probabilistic loops with non-polynomial functions builds upon Prob-solvable loops.
In contrast, our PCE-based approach, described in the following sections, is not limited to Prob-solvable and can be used in more general probabilistic loops.
The only requirement is that the loops satisfy the conditions in Section~\ref{sec:conditions}.

\subsection{Polynomial Chaos Expansion}\label{sec:PCE}

Polynomial chaos expansion (PCE) \cite{Ernstetal2012,XiuKarniadakis2002a}
recovers a random variable in terms of a linear combination of functionals whose entries are known random variables, sometimes called germs, or, basic variables. 
Let $(\Omega, \Sigma, \pr)$ be a probability space, where $\Omega$ is the set of elementary events, $\Sigma$ is a $\sigma$-algebra of subsets of $\Omega$, and $\pr$  is a probability measure on $\Sigma$.
Suppose $X$ is a real-valued random variable defined on $(\Omega, \Sigma, \pr)$, such that 
\begin{align}\label{l2fcn}
\E(X^2) &=\int_{\Omega} X^2(\omega) d\pr(\omega) < \infty.
\end{align}
The space of all random variables $X$  satisfying \eqref{l2fcn} is denoted by $L^2(\Omega, \Sigma, \pr)$. That is, the elements of $L^2(\Omega, \Sigma, \pr)$  are real-valued random variables defined on $(\Omega, \Sigma, \pr)$ with finite second moments. If we define  the inner product as
 $   \E(XY)=(X,Y)=\int_{\Omega} X(\omega) Y(\omega) d\pr(\omega)$ 
and norm $||X||=\sqrt{\E(X^2)}=\sqrt{\int_{\Omega} X^2(\omega) d\pr(\omega)}$, then $L^2(\Omega, \Sigma, \pr)$ is a Hilbert space; i.e., an infinite dimensional linear space of functions endowed with an inner product and a distance metric.  
 Elements of a Hilbert space can be uniquely identified by their coordinates with respect to an orthonormal basis of functions, in analogy with Cartesian coordinates in the plane. Convergence with respect to the norm $||\cdot||$ is called \textit{mean-square convergence}. A particularly important feature of a Hilbert space is that when the limit of a sequence of functions exists, it belongs to the space.

The elements in $L^2(\Omega, \Sigma, \pr)$ can be classified in two groups:  \textit{basic}  and   \textit{generic} random variables, which we want to decompose using the elements of the first set of basic variables. \cite{Ernstetal2012} showed that the basic random variables that can be used in the decomposition of other functions have finite moments of all orders with continuous probability density functions (pdfs). 

The $\sigma$-algebra generated by the basic random variable $Z$ is denoted by $\sigma(Z)$.
Suppose we restrict our attention to decompositions of a random variable $X=g(Z)$, where $g$ is a function with  $g(Z) \in L^2(\Omega, \sigma(Z),\pr)$, and the basic random variable $Z$ determines the class of orthogonal polynomials $\{\phi_i(Z), i \in \mathbb{N}\}$,
\begin{align}
\left<\phi_i(Z),\phi_j(Z)\right>&= \int_{\Omega} \phi_i(Z(\omega))\phi_j(Z(\omega)) d\pr(\omega) 
=\int \phi_i(x)\phi_j(x) f_Z(x) dx=\begin{cases} 1 & i=j \\
0 & i \ne j \end{cases} \label{orth.poly}
\end{align}
where $f_Z$ denotes the pdf of $Z$. The set $\{\phi_i(Z), i \in \mathbb{N}\}$ is a polynomial chaos basis.

If $Z$ is normal with mean zero, the Hilbert space $L^2(\Omega, \sigma(Z),\pr)$ is called \emph{Gaussian} and the related set of polynomials is represented by the family of Hermite polynomials (see, for example, \cite{XiuKarniadakis2002a}) defined on the whole real line. Hermite polynomials form a basis of $L^2(\Omega, \sigma(Z),\pr)$. Therefore, every random variable $X$ with  finite second moment can be approximated  by the truncated PCE 
\begin{align}\label{trunc.PCE}
    X^{(d)}&=\sum_{i=0}^d c_i \phi_i(Z),
\end{align}
for suitable coefficients $c_i$ that depend on the random variable $X$. The truncation parameter $d$ is the highest polynomial degree in the expansion. Since the polynomials are orthogonal,
\begin{align}
    c_i&=\frac{1}{||\phi_i||^2} \left< X,\phi_i \right> =\frac{1}{||\phi_i||^2}\left<g,\phi_i\right> 
    =\frac{1}{||\phi_i||^2} \int_{\real} g(x)\phi_i(x) f_Z(x) dx.
\end{align}
The truncated PCE of $X$ in \eqref{trunc.PCE} converges in mean square to $X$ \cite[Sec. 3.1]{Ernstetal2012}. The first two moments of \eqref{trunc.PCE} are determined by
 $   \E(X^{(d)})= c_0,$ and $ \var(X^{(d)})= \sum_{i=1}^d c_i^2 ||\phi_i||^2.$
 
Representing a random variable by a series of Hermite polynomials in a countable sequence
of independent Gaussian random variables is  known as Wiener–Hermite polynomial chaos expansion.  In applications of Wiener–Hermite PCEs, the underlying Gaussian Hilbert space is
often taken to be the space spanned by a  sequence $\{Z_i, i \in \mathbb{N}\}$ of independent standard Gaussian basic random variables; i.e.,  $Z_i \sim \Ncal(0, 1)$. For computational purposes, the
countable sequence $\{Z_i, i \in \mathbb{N}\}$ is restricted to a finite number $k \in \mathbb{N}$  of random variables. The Wiener–Hermite PCE converges for 
random variables with finite second moment. Specifically, for any random variable $X \in L^2(\Omega, \sigma(\{Z_i, i \in \mathbb{N}\}), \pr)$, the approximation \eqref{trunc.PCE} satisfies
\begin{align}\label{convergence}
    X_k^{(d)}=\sum_{i=0}^d a_i \phi_i(Z_1,\ldots,Z_k) &\to X \quad \mbox{as } d,k \to \infty
\end{align} 

in mean-square convergence (see \cite{Ernstetal2012}). 
The distribution of $X$ can be quite general; e.g., discrete, singularly continuous, absolutely continuous as well as of mixed type.

\section{Polynomial Chaos Expansion Algorithm}\label{sec:PCE_algo}
\subsection{Random Function Representation}\label{sec:conditions}

In this section, we state the conditions under which the estimated polynomial is an unbiased and consistent estimator and has exponential convergence rate.
Suppose $k$ continuous random variables  $Z_{1},\ldots, Z_{k}$ are used to introduce stochasticity in a PP with corresponding cumulative distribution functions (cdfs) $F_{Z_i}$, $i=1,\ldots,k$. Also, suppose all $k$ distributions have probability density functions, and let  $\Z=(Z_{1},\ldots,Z_{k})$ with cdf $F_{\Z}$. We assume that the elements of $\Z$ satisfy the following conditions:

\begin{itemize}
    \item[(A)] $Z_{i}$, $i=1,\ldots, k$, are independent.
    \item[(B)] We consider functions $g$ such that $g(\Z) \in L^2(\mathcal{Q}, F_\Z)$, where $\mathcal{Q}$ is the  support of the joint distribution of  $\Z=(Z_{1},\ldots,Z_{k})$.\footnote{$\Omega$ is dropped from the notation as the sample space is not important in our formulation.}
    \item[(C)] All random variables $Z_{i}$ have distributions that are uniquely characterized by their moments.\footnote{Conditions that ascertain this are given in Theorem 3.4 of \cite{Ernstetal2012}.}
\end{itemize}

Under condition (A),  the joint cdf of the components of $\Z$ is $F_{\Z}= \prod_{i=1}^k F_{Z_i}$. 
To ensure the construction of unbiased estimators with optimal exponential convergence rate (see \cite{XiuKarniadakis2002a}, \cite{Ernstetal2012}) in the context of probabilistic loops, we further introduce the following assumptions:

\begin{itemize}
    \item[(D)] $g$ is a function of a fixed number of basic variables (arguments) over all loop iterations.
    \item[(E)] If $\Z(j)=(Z_{1}(j),\ldots,Z_{k}(j))$ is the stochastic argument of $g$ at iteration $j$, then $F_{Z_i(j)}(x) = F_{Z_{i}(l)}(x)$ for all pairs of iterations $(j,l)$ and $x$ in the support of $F_{Z_i}$.
\end{itemize}

If Conditions (D) and (E) are not met, then the polynomial coefficients in the PCE need be computed for each loop iteration individually to ensure optimal convergence rate.
It is straightforward to show the following proposition.

\begin{proposition}\label{propos::func_props}
If $\Z=(Z_{1},\ldots, Z_{k_1})$, $\Y=(Y_{1},\ldots, Y_{k_2})$ satisfy conditions (A), (C) and (E) and are mutually independent, and functions $g$ and $h$ satisfy conditions (B) and (D),  then their sum, $g(\Z) + h(\Y)$, and product, $g(\Z) \cdot h(\Y)$, also satisfy conditions (B) and (D).
\end{proposition}

\begin{figure}[t!]
\resizebox{10cm}{10cm}{%
\hspace{-0.5cm}
\begin{tikzpicture}[node distance=2cm]
\node (start) [startstop] {Start};
\node (in1) [io] {\textbf{Input:}
    $\begin{cases}
      Z_{1}, \hspace{0.5cm} \bar{d}_{1},\\
      Z_{2}, \hspace{0.5cm} \bar{d}_{2},\\
      ...\\
      Z_{k}, \hspace{0.5cm} \bar{d}_{k},\\
    \end{cases}$\vspace{0.5cm}
    
Set of random variables and the highest degrees of the corresponding  univariate orthogonal polynomials

};
    
\node (in2) [io, right of=in1, xshift = 4cm] {\textbf{Input}
    $\begin{cases}
      f_{1}(z_{1}), \hspace{0.5cm} \left[a_{1}, b_{1}\right],\\
      f_{2}(z_{2}), \hspace{0.5cm} \left[a_{2}, b_{2}\right],\\
      ...\\
      f_{k}(z_{k}), \hspace{0.5cm} \left[a_{k}, b_{k}\right],\\
    \end{cases}$\vspace{0.5cm}
Set of the density functions and the corresponding supports
};

\node (in3) [io, right of=in2, xshift = 4cm] {\textbf{Input} \\\vspace{0.1cm}
$g(z_{1}, ..., z_{k})$\vspace{0.5cm}

Target function

};

\node (pro1) [process, below of=in2, yshift = -3cm] {\textbf{Orthonormal polynomials:}\\
\vspace{0.2cm}
$\begin{cases}
    \{\bar{p}_{1}^{deg}\}_{deg=0}^{\bar{d}_{1}}\\
    ...\\
    \{\bar{p}_{k}^{deg}\}_{deg=0}^{\bar{d}_{k}}\\
\end{cases}$

};

\node (pro2) [process, below of=in1, yshift = -12.7cm] {
\[
\D^{L \times k} = 
\begin{blockarray}{ccccc}
Z_{1} & Z_{2} & Z_{3} & ... & Z_{k} \\
\begin{block}{(ccccc)}
  0 & 0 & 0 & 0 & 0 \\
  1 & 0 & 0 & 0 & 0 \\
  \vdots & \vdots & \vdots  & \ddots & \vdots\\
  \bar{d}_{1} & \bar{d}_{2} & \bar{d}_{3} & \cdots & \bar{d}_{k} \\
\end{block}
\end{blockarray}
 \]
};

\draw [arrow, line width=1mm, color = green!50, text = black] (in2) -- node[anchor=east] {\begin{center}
    \textbf{1: Gram-Schmidt Process}
\end{center}} (pro1);
\draw [arrow, line width=1mm, color = green!90] (in1) |-  (pro1);

\begin{scope}[transform canvas={xshift=-1cm}]
  \draw [arrow, line width=1mm, color = brown!90, text = black] (in1) -- node[anchor=east] {\vspace{5cm}
\begin{center}
    \textbf{2:  Matrix of\\  polynomials' combinations}
\end{center}: } (pro2);
\end{scope}

\node (pro3) [process, below of=in3, yshift = -6.5cm] {\textbf{Fourier coefficients}:\\
$c_{j} = \int\limits_{a_{1}}^{b_{1}}...\int\limits_{a_{k}}^{b_{k}}g(z_{1}, ...z_{k})\prod\limits_{i=1}^{k}\left[f_{Z_{i}}(z_{i})\bar{p}^{d_{ji}}_{i}(z_{i})\right]dz_{k}...d_{z_{1}}$

};

\draw [arrow, line width=1mm, color = blue!90, text = black] (in3) --  node[anchor=west] {\vspace{3cm}
\begin{center}
    \textbf{3: \textbf{Calculation of\\ coefficients}}
\end{center} } (pro3);

\begin{scope}[transform canvas={xshift=1cm}]
  \draw [arrow, line width=1mm, color = blue!90, text = black, shorten <= -1cm] (pro2) -| (pro3);
\end{scope}

\draw [arrow, line width=1mm, color = blue!90, text = black] (pro1) -- (pro3);

\node (out1) [io, below of=pro1, yshift=-5cm] {\textbf{Output:} \\
\vspace{0.2cm}
$\sum\limits_{j} c_{j}\prod\limits_{i} \bar{p}_{i}^{d_{ji}}$
};

\draw [arrow, line width=1mm] (pro3) |- (out1);

\draw [arrow, line width=1mm] (pro1) -- node[anchor=west] {\vspace{3cm}
\begin{center}
    \textbf{4: \textbf{Summation of weighted polynomials}}
\end{center} } (out1);
\end{tikzpicture}
}
\caption{Illustration of PCE algorithm}
\label{fig::PCE_Algo}
\end{figure}
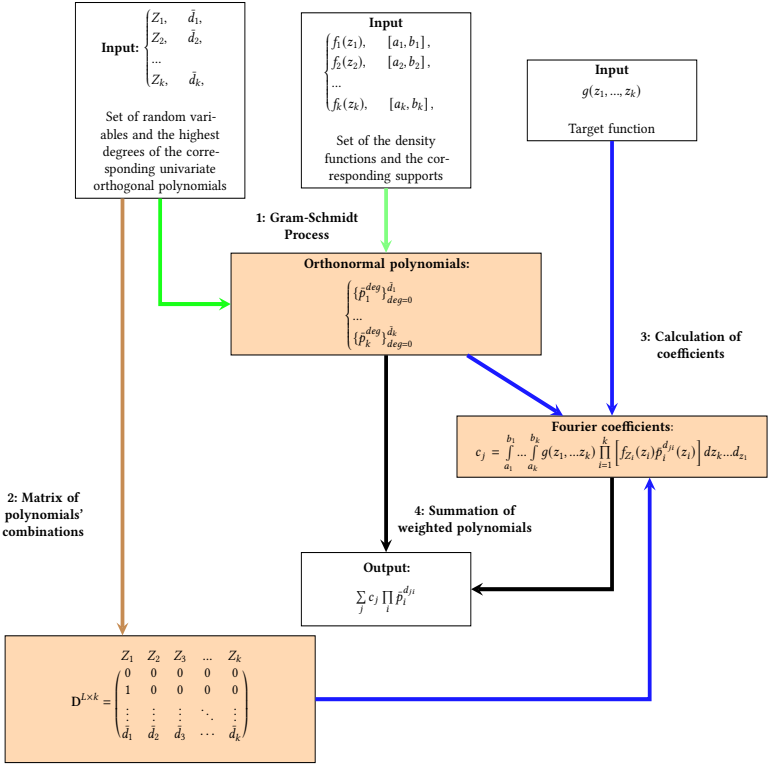


\subsection{PCE Algorithm}\label{sec:PCEalgo}

Let $Z_1,\ldots,Z_k$ be independent continuous random variables, with respective cdfs $F_i$, satisfying conditions (A), (B) and (C). Then, $\Z=(Z_1,\ldots,Z_k)^T$ has cdf $\F_{\Z} = \prod_{i=1}^k F_{i}$. Let $\mathcal{Q}$ denote the support of $\F_{\Z}$. The function $g: \real^{k} \to \real$, with $g \in L^2(\mathcal{Q}, \F)$, can be approximated  with the truncated orthogonal polynomial expansion, as described in  Fig. \ref{fig::PCE_Algo}, 
\begin{align}\label{polyexpanse}
g(\Z) &\approx \hat{g}(\Z) = \sum_{\substack{d_i \in \{0,\ldots,\bar{d}_i\},\\ i=1,\ldots, k}} c(d_1,\ldots,d_k) z_1^{d_1}  \cdots z_k^{d_k} 
= \sum_{j=1}^L c_{j}\prod\limits_{i=1}^{k}\bar{p}_{i}^{d_{ji}}(z_{i}), 
\end{align}
where 
\begin{itemize}
 \item $\bar{p}_{i}^{d_{ji}}(z_{i})$ is a polynomial of degree $d_{ji}$, and belongs to the set of orthogonal polynomials with respect to $F_{Z_i}$ that are calculated with the Gram-Schmidt orthogonalization procedure\footnote{Generalized PCE typically entails using orthogonal basis polynomials specific to the distribution of the basic variables, according to the Askey scheme of \cite{XiuKarniadakis2002a,Xiu2010}. We opted for the most general procedure that can be used for any basic variable distribution.};
    \item $\bar{d}_{i} = \max\limits_{j}(d_{ji})$ is the highest degree of the univariate orthogonal polynomial, for $i=1,\ldots,k$; 
    \item $L = \prod\limits_{i=1}^{k}(1 + \bar{d}_{i})$ is the total number of multivariate orthogonal polynomials and equals the truncation constant;
    \item $c_{j}$ are the Fourier coefficients.
\end{itemize}

The Fourier coefficients are calculated using
\begin{align}\label{coef}
c_{j} &= \int\limits_{\mathcal{Q}}g(z_{1},...,z_{k})p_{i}^{d_{ji}}(z_{i})d\F = 
\idotsint\limits_{\mathcal{Q}} g(z_{1},...,z_{k})\left(\prod\limits_{i=1}^{k}\bar{p}_{i}^{d_{ji}}(z_{i})\right) dF_{Z_k}...dF_{Z_1},
\end{align}
by Fubini's theorem. 

\begin{example}
Returning to the Turning vehicle model in Fig.~\ref{fig:turningexample}, 
the non-polynomial functions to approximate are $g_1 = cos$ and $g_2 = sin$ from the updates of program variables $y, x$, respectively. In both cases, we only need to consider a single basic random variable, $Z \sim \Ncal(0,0.01)$ ($\psi$ in Fig.~\ref{fig:turningexample}). 

Using polynomials of degree up to $5$,
\eqref{polyexpanse} has the following form for the two functions,
\begin{equation}\label{eqn:g_1}
    \hat{g}_1(z) = cos(\psi) = a_0 + a_1\psi + ... + a_5\psi^5
\end{equation}
and
\begin{equation}\label{eqn:g_2}
    \hat{g}_2(z) = sin(\psi) = b_0 + b_1\psi + ... + b_5\psi^5.
\end{equation}
We compute the coefficients $a_i, b_i$ in equations~\eqref{eqn:g_1}-\eqref{eqn:g_2} using~\eqref{coef} to obtain the values shown in Fig.~\ref{fig:turningexample}.
\end{example}

\begin{example}
Consider the function $g(x, y) = \sqrt{x^{2} + y^{2}}$, where $x \sim Uniform(-1, 1)$ and $y$ has pdf $f_{y} = 0.75(1 - y^{2})$ supported on  $\left[-1, 1\right]$.
Up to degree 2, the basis elements in the PC expansion are element-wise products of the univariate orthogonal polynomials
\begin{gather*}
    p^{0}_{x}(x) = 1, \quad  p^{0}_{y}(y) = 1,
    p^{1}_{x}(x) = \sqrt{3}x, \quad  p^{1}_{y}(y) = \sqrt{5}y,\\
    p^{2}_{x}(x) = 1.5\sqrt{5}x^{2} - 0.5\sqrt{5}, \,   p^{2}_{y}(y) = 1.25\sqrt{14}y^{2} - 0.25\sqrt{14}.
\end{gather*}
The corresponding  PCE polynomial basis elements are  
\begin{gather*}
    p^{00}_{xy}(x, y)  = 1, \quad 
    p^{02}_{xy}(x, y) = 1.25\sqrt{14}y^{2} - 0.25\sqrt{14}, \quad 
    p^{20}_{xy}(x, y) = 1.5\sqrt{5}x^{2} - 0.5\sqrt{5},\\
    p^{22}_{xy}(x, y) = 1.875\sqrt{70}x^{2}y^{2} - 0.375\sqrt{70}x^{2} - 0.625\sqrt{70}y^{2} + 0.125\sqrt{70},
\end{gather*}
with corresponding non-zero Fourier coefficients 
$c_{00} = 0.677408, c_{02} = 0.154109, c_{20} = 0.216390$, and $c_{22} = -0.040153.$ 
The resulting estimator is 
\begin{align}\notag
    \hat{g}(x, y) = \sum\limits_{(i, j) = (0, 0)}^{(2, 2)}c_{i, j}p^{ij}_{xy}(x, y) = -0.629900x^{2}y^{2} + 0.851774x^{2} + 0.930747y^{2} + 0.249327.
\end{align}
\end{example}


\paragraph{Complexity.} 

Assuming the expansion is carried out up to the same polynomial degree $d$ for each basic variable, $\bar{d}_{i} = d$, $\forall i=1,...,k$. This implies $d=\sqrt[k]{L}-1$. The complexity of the scheme is $\mathcal{O}(sd^{2}k + s^{k}d^{k})$, where $\mathcal{O}(s)$ is the complexity of computing univariate integrals.

The complexity of our approximation scheme consists of of two parts:
(1) the orthogonalization process and (2) the calculation of coefficients. Regarding (1), we orthogonalize and normalize $k$ sets of $d$  basic linearly independent polynomials via the Gram-Schmidt process.
For degree $d{=}1$, we need to calculate one integral, the inner product with the previous polynomial.
Additionally, we need to compute one more integral, the norm of itself (for normalization).
For each subsequent degree $d'$, we must calculate $d'$ additional new integrals. The computation of each integral has complexity $\mathcal{O}(s)$.
Regarding (2), the computation of the coefficients requires  calculating $L{=}(d{+} 1)^{k}$ integrals with $k$-variate functions as integrands.

We define the approximation error  to be
\begin{equation}\label{error_std}
    se(\hat{g}) = \sqrt{\int\limits_{\mathcal{Q}}\left(g(z_{1},...,z_{k}) - \hat{g}(z_{1},...,z_{k})\right)^{2}dF_{Z_1}\ldots dF_{Z_k}}
\end{equation}
since $\E(\hat{g}(Z_{1},...,Z_{k}))=g(Z_1,\cdots,Z_k)$ by construction.

The implementation of this algorithm may become challenging when the random functions  have complicated forms and  the number of parametric uncertainties is large. In this case, the calculation of the PCE coefficients involves high dimensional integration, which may prove difficult and time prohibitive for real-time applications \cite{SonDu2020}.

\section{Prob-Solvable Loops for General Non-Polynomial Functions}\label{sec:genfun}

PCE\footnote{We provide further details about PCE computation in Appendix~2 in \cite{KofnovMSBB22}.}
allows incorporating non-polynomial updates into Prob-solvable loop programs and use the algorithm in \cite{Bartoccietal2019} and exact tools, such as \textsc{Polar}~\cite{Moosbruggeretal2022}, for moment (invariant) computation.
We identify two classes of programs based on how the distributions of the generated random variables vary.

\subsection{Iteration-Stable Distributions of Random Arguments}\label{sec:stable} 

Let $\Pcal$ be an arbitrary Prob-solvable loop and suppose that a (non-basic) state variable $x \in \Pcal$ has a non-polynomial $L^2$-type update $g(\Z)$, where $\Z = (Z_{1}, \ldots, Z_{k})^{T}$ is a vector of (basic) continuous, independent, and identically distributed random variables \textit{across iterations}. That is, if $f_{Z_j(n)}$ is the pdf  of the random variable $Z_j$ in iteration $n$, then  $f_{Z_j( n)} \equiv f_{Z_j(n')}$, for all iterations $n, n'$ and $j=1,\ldots,k$. The basic random variables $Z_1,\ldots, Z_k$ and the update function $g$ satisfy conditions (A)--(E) in Section~\ref{sec:conditions}.
For the class of Prob-solvable loops where all variables with non-polynomial updates satisfy these conditions,  the computation of the Fourier coefficients in the PCE approximation \eqref{polyexpanse} can be carried out as explained in Section~\ref{sec:PCEalgo}.
In this case, the convergence rate is optimal.  

\subsection{Iteration Non-Stable Distribution of Random Arguments}\label{sec:non-stable} 

Let $\Pcal$ be an arbitrary Prob-solvable loop and suppose that a state variable $x \in \Pcal$ has a non-polynomial $L^2$-type update $g(\Z)$, where $\Z = (Z_{1}, \ldots, Z_{k})^{T}$ is a vector of continuous independent  but \textit{not necessarily identically} distributed random variables across iterations. 
For this class of Prob-solvable loops, conditions (A)--(C) in Section~\ref{sec:conditions} hold, but (D) and/or (E) may not be fulfilled.
In this case, we can ensure optimal exponential convergence by fixing the number of loop iterations.
For unbounded loops, we describe an approach converging in mean-square and establish its convergence rate next.

\paragraph{Conditional estimator given number of iterations.}

Let $N$ be an a priori fixed finite integer, representing the maximum iteration number. The set $\{1, \ldots, N\}$ is a finite sequence of iterations for the Prob-solvable loop $\Pcal$.

Iterations are executed sequentially for $n =1,\ldots, N$, which allows  the estimation of  the final functional that determines the target state variable at each iteration $n \in \{1, \ldots, N\}$ and its set of supports. Knowing these features, we can carry out $N$ successive expansions. Let $P(n)$ be a~PCE of~$g(\Z)$ for iteration $n$. We introduce an additional program variable $c$ that counts the loop iterations.
The variable $c$ is initialized to $0$ and incremented by $1$ at the beginning of every loop iteration.
The final estimator of  $g(\Z)$ can be represented as 
\begin{equation}\label{eq:cond-estimator}
    \hat{g}(\Z) = \sum\limits_{n = 1}^{N}P(n)\left[\prod\limits_{j = 1, j \neq n}^{N}\frac{(c - j)}{n - j}\right].
\end{equation}
Replacing non-polynomial functions with \eqref{eq:cond-estimator} results in a program with only polynomial-type updates and \textit{constant} polynomial structure; that is, polynomials with  coefficients that remain constant across iterations. Moreover, the estimator is unbiased with optimal exponential convergence on the set of iterations $\{1, \ldots, N\}$ \cite{XiuKarniadakis2002a}.

\paragraph{Unconditional estimator.}\label{sec:uncond} 
Here the iteration number is unbounded. Without loss of generality, we consider a single basic random variable $Z$; that is, $k{=}1$.  The function  $g(Z)$ is scalar-valued and can be represented as a polynomial of \textit{nested} $L^2$ functions, which depend on polynomials of the argument variable. Each nested functional argument is expressed as a sum of orthogonal polynomials yielding  the final estimator, which is itself a polynomial.

Since PCE converges to the function it approximates in mean-square (see \cite{Ernstetal2012}) 
on the whole interval (argument's support),  PCE converges on any sub-interval of the support of the argument in the same sense.

Let us consider a function $g$ with a sufficiently large domain and a random variable $Z$ with known distribution and support. For example, $g(Z) = e^{Z}$, with $Z \sim N(\mu, \sigma^{2})$. The domain of $g$ and the support of $Z$ are the real line.
We can expand $g$ into a PCE with respect to the distribution of  $Z$ as 
\begin{align}\label{truef}
 g(Z) &= \sum\limits_{i=0}^{\infty}c_{i}p_{i}(Z).
\end{align}
The distribution of $Z$ is reflected in the polynomials in \eqref{truef}. Specifically,  $p_{i}$, for $i=0,1,\ldots$, are  Hermite polynomials of special type in that they are orthogonal (orthonormal) with respect to $N(\mu, \sigma^{2})$. They also form an  orthogonal basis of the space of $L^2$ functions. Consequently, any function in $L^2$ can be estimated arbitrarily closely by these polynomials.
In general, any continuous distribution with finite moments of all orders and sufficiently large support can also be used as a model for basic variables in order to construct a basis for $L^2$ (see \cite{Ernstetal2012}).

Now suppose that the distribution of the  underlying variable $Z$ is unknown with pdf $f_Z(z)$ that is continuous on its support $\left[a, b \right]$. Then, there exists another basis of polynomials, $\{q_{i}\}_{i=0}^{\infty}$, which are orthogonal on the support $\left[a, b \right]$ with respect to the pdf $f_Z$. Then, on the interval $\left[a, b \right]$,  $g(Z) = \sum_{i=0}^{\infty}k_{i}q_{i}(Z)$, and $\mathbb{E}_{f}\left[g(Z)\right] = \mathbb{E}_{f}\left[\sum_{i=0}^{M}k_{i}q_{i}(Z) \right]$, $\forall M \geq 0$.

Since  $\left[a, b \right] \subset \real$,  the expansion $\sum_{i=0}^{\infty}c_{i}p_{i}(Z)$ converges  in mean-square to $g(Z)$ on $\left[a, b \right]$. In the limit, we have  $g(Z) = \sum_{i=0}^{\infty}c_{i}p_{i}(Z)$ on the interval $\left[a, b \right]$. Also, $\mathbb{E}_{f}\left(g(Z)\right) = \mathbb{E}_{f}\left(\sum_{i=0}^{\infty}c_{i}p_{i}(Z)\right)$ for the true pdf $f$ on $\left[a, b \right]$.
In general, though, it is not true that $\mathbb{E}_{f}\left(g(Z)\right) = \mathbb{E}_{f}\left(\sum_{i=0}^{M}c_{i}p_{i}(Z)\right)$ for any arbitrary $M \geq 0$ and any pdf $f_Z(z)$ on $\left[a, b \right]$, as the estimator is biased.

To capture this discrepancy, we define the approximation error as
\begin{align}\label{approxer}
  e(M) &= \mathbb{E}_{f}\left( g(Z) -  \sum\limits_{i=0}^{M}c_{i}p_{i}(Z) \right)^{2} = \mathbb{E}_{f}\left(\sum\limits_{i=M+1}^{\infty}c_{i}p_{i}(Z) \right)^{2}.
\end{align}
\paragraph{Computation of error bound.}

Assume the true pdf $f_Z$ of $Z$ is supported on $\left[ a, b\right]$. Also, assume the domain of $g$ is $\mathbb{R}$. The random function $g(Z)$ has PCE  on the whole real line based on Hermite polynomials $\{p_{i}(Z)\}_{i=0}^{\infty}$ that are orthogonal with respect to the standard normal pdf $\phi$. The truncated  expansion estimate of \eqref{truef} with respect to a normal basic random variable is
\begin{align}
    \hat{g}(Z)&=  \sum\limits_{i=0}^{M}c_{i}p_{i}(Z). \label{fapprox}
\end{align} 
We compute an upper bound for the approximation error for our scheme in Theorem~\eqref{err-bound}.

\begin{theorem}\label{err-bound}
Suppose $Z$ has density $f_Z$ supported on $[a,b]$, $g: \real \to \real$ is in $L^2$, and $\phi$ denotes the standard normal pdf. Under \eqref{truef} and \eqref{fapprox},
\begin{align}
   \left \| g(Z) - \hat{g}(Z) \right \|_{f}^{2}&= \int_a^b \left(g(z)-\hat{g}(z)\right)^2 f_Z(z) dz \notag \\
   & \le   \left( \frac{2}{\min{(\phi(a), \phi(b))}} + 1 \right) \var_{\phi}\left(g(Z)\right). \label{bound}
\end{align}
\end{theorem}

\medskip
\noindent 
The upper bound in \eqref{bound} depends only on the support of $f_Z$ and the function $g$. 
If $Z$ is standard normal ($f_Z=\phi$), then the upper bound in \eqref{bound} equals $\var_\phi(g(Z))$. We provide the proof of Theorem~\ref{err-bound} in Appendix~\ref{app:A}.

\begin{remark}
The approximation error inequality in \cite[Lemma 1]{Muehlpfordtetal2018}, 
\begin{equation}\label{muelpford}
\left \|{ g(Z) - \sum _{i=0}^{T} c_{i}  p_{i}(Z) }\right \| \leq \frac {\| g(Z)^{(k)} \|}{ \prod _{i=0}^{k-1} \sqrt {T-i+1}}, 
\end{equation}
is a special case of  Theorem \ref{err-bound} when $Z \sim \mathcal{N}(0,1)$ and $f_Z=\phi$, and the  polynomials $p_i$ are Hermite. In this case, the left hand side of \eqref{muelpford} equals  
$\sqrt{\sum_{i=n+1}^{\infty}c_{i}^{2}}$.
\end{remark}

Although Theorem~\ref{err-bound} is restricted to distributions with bounded support, the approximation in  \eqref{fapprox}  also converges for distributions with unbounded support.

\section{Exact Moment Derivation}\label{sec:exact}

In Sections \ref{sec:PCE_algo} and \ref{sec:genfun}, we combined PCE with the \textit{Prob-solvable loop} algorithm of \cite{Moosbruggeretal2022,Bartoccietal2019} to compute PCE \emph{approximations} of the moments of the distributions of the random variables generated in a probabilistic loop. 
In this section, we develop a method for the derivation of the \emph{exact} moments of probabilistic loops that comply with a specified loop structure and functional assignments.
\cite{Bartoccietal2019}, and later \cite{Moosbruggeretal2022}, introduce a technique for exact moment computation for \textit{Prob-solvable loops} without non-polynomial functions.
\textit{Prob-solvable loops} support common probability distributions $F$ with constant parameters and program variables $x$ with specific polynomial assignments (cf. Section~\ref{sec:preliminaries}).
In Section~\ref{subsec:trigexp}, we first show how to compute $\E(h(F)^k)$ where $h$ is the exponential or a trigonometric function.
In Section~\ref{subsec:exactvars}, we describe how to incorporate trigonometric and exponential updates of program variables $x$ into the \textit{Prob-solvable loop} setting.

\subsection{Trigonometric and Exponential Functions for Distributions}\label{subsec:trigexp}

The first step in supporting trigonometric and exponential functions in Prob-solvable loops is to understand how to compute the expected values of random variables that are trigonometric and exponential functions of random variables with known distributions.
Due to the polynomial arithmetic supported in Prob-solvable loops, non-polynomial functions of random variables can be mixed via multiplication in the resulting program.
We adopt the results from \cite{Jasouretal2021} providing a formula for the expected value of mixed trigonometric polynomials of distributions, given the distributions' characteristic functions.

\begin{definition}
We call  $p(x)$ a \emph{standard polynomial of order $d \in \nat$} if   
\begin{align*}
       p(x) = \sum_{k = 0}^{d}\alpha_{k}x^{k},
\end{align*}
with coefficients $\alpha_{k} \in \real$.
Further, $p(x)$ is defined to be a \emph{mixed trigonometric polynomial} if it is a mixture of a standard polynomial with trigonometric functions of the form 
\begin{align*}
    p(x) = \sum_{k=0}^{d}\alpha_{k}x^{b_{k}}\cdot cos^{c_{k}}(x)\cdot sin^{s_{k}}(x),
\end{align*}
with $b_{k}, c_{k}, s_{k} \in \nat$ and coefficients $\alpha_{k} \in \real$ \cite{Jasouretal2021}. 
\end{definition}

Following \cite{Jasouretal2021}, we define the mixed-trigonometric-polynomial moment of order $\alpha$ for a random variable $X$ as 
\begin{align}\label{mix-trig-mom}
m_{X^{\alpha_1} c_{X}^{\alpha_{2}} s_{X}^{\alpha_3}} &= \E \left[X^{\alpha_{1}} \cos^{\alpha_{2}}(X) \sin^{\alpha_{3}}(X)\right],
\end{align} 
where $(\alpha_{1}, \alpha_{2}, \alpha_{3}) \in \nat^{3}$ such that $ \alpha=\sum_{k=1}^{3}\alpha_{k} $.
When the characteristic function of the random variable $X$ is known, Lemma 4 in \cite{Jasouretal2021}, together with the linearity of the expectation operator, provides the computation rule for \eqref{mix-trig-mom}:
 \begin{align}
    m_{X^{\alpha_{1}} c_{X}^{\alpha_{2}} s_{X}^{\alpha_{3}}} = \frac{1}{i^{\alpha_{1} + \alpha_{3}}2^{\alpha_{2} + \alpha_{3}}} \times \sum\limits_{(k_{1}, k_{2}) = (0, 0)}^{(\alpha_{2}, \alpha_{3})}\binom{\alpha_{2}}{k_{1}}\binom{\alpha_{3}}{k_{2}}(-1)^{\alpha_{3} - k_{2}}\frac{d^{\alpha_{1}}}{dt^{\alpha_{1}}}\Phi_{X}(t)\bigg|_{t = 2(k_{1} + k_{2}) - \alpha_{2}-\alpha_{3}}.
\end{align}

\begin{example}
Let $\X \sim \mathcal{N}(0,1)$. Its characteristic function is $\Phi_X(t) = \exp(-0.5t^2)$ and $\Phi'_X(t)  = -t\exp(-0.5t^2)$. Then,
\begin{align*}
    \E\left[X \sin(X) \cos(X)\right] &= \frac{1}{i^2 2^2} \left( -\Phi'_X(-2) + \Phi'_X(0) - \Phi'_X(0) + \Phi'_X(2)  \right) = \exp(-2).
\end{align*}
\end{example}

To support exponential functions in Prob-solvable loops, we define  \emph{mixed exponential polynomials} as 
\begin{align*}
    p(x) = \sum_{k = 0}^{d}\alpha_{k}x^{b_k}\exp^{e_k}(x),
\end{align*}
with $b_k, e_k \in \nat$ and coefficients $\alpha_k \in \real$.
Lemma \ref{lem:mom-mix} obtains a computational rule for moments of mixed exponential polynomials, provided they exist, in terms of the moment-generating function of the random variable $X$.

\begin{lemma}\label{lem:mom-mix}
Let $X$ be a random variable with moment-generating function $\mathcal{M}_{X}(t) = \E\left[\exp(t X)\right]$. Suppose $\alpha_1, \alpha_2 \in \nat$, and let $\alpha = \alpha_1 + \alpha_2$.
If the  mixed-exponential-polynomial moment of order $\alpha$ exists, it  can be computed using the following formula
\begin{align}
    m_{X^{\alpha_1} e^{\alpha_2}_X} &=\E\left[X^{\alpha_1}\exp^{\alpha_2}(X)\right]= \frac{d^{\alpha_1}}{dt^{\alpha_1}} \mathcal{M}_{X}(t) \bigg|_{t=\alpha_2}.
\end{align}
\end{lemma}
\begin{proof}
\begin{align*}
    m_{X^{\alpha_1} e^{\alpha_2}_X} = 
    \E \left[X^{\alpha_1}\exp^{\alpha_2}(X) \right] =
    \E \left[\frac{d^{\alpha_1}}{dt^{\alpha_1}}\exp(t X) \bigg|_{t=\alpha_2} \right] = 
    \frac{d^{\alpha_1}}{dt^{\alpha_1}} \E \left[\exp(t X) \right]  \bigg|_{t=\alpha_2} = 
    \frac{d^{\alpha_1}}{dt^{\alpha_1}} \mathcal{M}_{X}(t) \bigg|_{t=\alpha_2}.
\end{align*}
\end{proof}

\begin{example}
Suppose $X \sim \mathcal{N}(0,1)$ with moment generating function $\mathcal{M}_X(t) = \exp(0.5t^2)$.
Since  $X \exp^2(X)$ is integrable with respect to the normal pdf, 
\begin{align*}
\E\left[ X \exp^2(X)\right] = \frac{d}{dt} \exp(0.5t^2) \bigg|_{t=2} = t \cdot \exp(0.5t^2) \bigg|_{t=2} = 2 \cdot \exp(2).
\end{align*}
\end{example}

\subsection{Trigonometric and Exponential Functions in Variable Updates}\label{subsec:exactvars}

We now examine the presence of trigonometric and exponential functions of program variables, specifically of \emph{accumulator} variables, in Prob-solvable loops.

\begin{definition}[Accumulator]
    We call a program variable $x$ an \emph{accumulator} if the update of $x$ in the loop body has the form $x = x + z$, such that $z_i$ and $z_{i+1}$ are independent and identically distributed for all $i \geq 1$.  
\end{definition}

Consider a loop~$\Pcal$ with an accumulator variable $x$, updated as $x = x + z$, and a trigonometric or exponential function $h(x)$. 
Further, assume that the characteristic function (if $h=\sin$ or $h=\cos$) or the moment-generating function (if $h=\exp$) of $z$ is known.
Note that the distribution of the variable $x$ is, in general, different in every iteration.
Listing~\ref{lst:protoex} gives an example of such a loop.

\begin{lstlisting}[mathescape,label={lst:protoex}, caption={PP loop prototype}]
$\displaystyle z = 0;\; x = 0;\; y = 0$
while true:
     $\displaystyle z = Normal(0,1)$
     $\displaystyle x = x + z$
     $\displaystyle y = y + h(x)$ 
end
\end{lstlisting}

The idea now is to transform $\Pcal$ into an equivalent Prob-solvable loop $\Pcal'$ such that the term~$h(x)$ does not appear in $\Pcal'$.
In the following, we assume, for simplicity, that first $z$ is updated in $\Pcal$, then $x$, and only then $h$ is used. 
The following arguments are analogous if the updates are ordered differently (only the indices change).
Note that we can rewrite $h(x_{t+1})$ as $h(x_{t} + z_{t+1})$. 

\paragraph{Transforming $\exp(x)$.}
In the case of $h=\exp$, we have 
\begin{equation}
\exp(x_{t+1})
= \exp(x_{t} + z_{t+1})
= \exp(x_{t}) \exp(z_{t+1}).
\end{equation}
We utilize this property and transform the program $\Pcal$ into a program $\Pcal'$ by introducing an auxiliary variable $\hat x$ that models the value of $\exp(x)$.
The update of $\hat x$ in the loop body succeeds the update of $x$ and is
\begin{equation}
\hat x = \hat x \cdot \exp(z).\label{eq:upd-exp}
\end{equation}
The auxiliary variable is initialized as $\hat x_{0} = \exp(x_0)$.
We then replace $h(x)$ by $\hat x$ in $\Pcal$ to arrive at our transformed program $\Pcal'$.
Because $z$ is identically distributed in every iteration and its moment-generating function is known, we can use the results from Section~\ref{subsec:trigexp} to compute any moment of $\exp(z)$.
Thus, the update in \eqref{eq:upd-exp} is supported by Prob-solvable loops.

\paragraph{Transforming sine and cosine.}
In the case of $h=\sin$ or $h=\cos$, 
applying standard trigonometric identities obtains
\begin{align}
\begin{split}
\cos(x_{t+1})
&= \cos(x_{t} + z_{t+1})
= \cos(x_{t})\cos(z_{t+1}) - \sin(x_{t})\sin(z_{t+1}),\\
\sin(x_{t+1})
&= \sin(x_{t} + z_{t+1})
= \sin(x_{t})\cos(z_{t+1}) + \cos(x_{t})\sin(z_{t+1}).
\end{split}
\end{align}
We introduce two auxiliary variables, $s$ and $c$, modeling the values of $\sin(x)$ and $\cos(x)$, simultaneously updated\footnote{Simultaneous updates $c, s = \text{expr}_1, \text{expr}_2$ can always be expressed as sequentially: $t_1 = \text{expr}_1; t_2 = \text{expr}_2; c = t_1; s = t_2$.} in the loop body as
\begin{equation}
    c,\; s = c \cdot \cos(z) - s \cdot \sin(z),\; s \cdot \cos(z) + c\cdot \sin(z),\label{eq:upd-trig}
\end{equation}
with the initial values $s_0 = \sin(x_0)$ and $c_0 = \cos(x_0)$. 
We then replace $h(x)$ with $s$ or $c$ in $\Pcal'$.
Again, because $z$ is identically distributed in every iteration and its characteristic function is known, we can use the results from Section~\ref{subsec:trigexp} to compute any moment of $\sin(z)$ and $\cos(z)$.
Thus, the update in \eqref{eq:upd-trig} is supported in Prob-solvable loops.

\begin{example}
Listings~\ref{lst:protoex2-trig} and~\ref{lst:protoex2-exp} show  the program from Listing~\ref{lst:protoex} with $h=\exp$ and $h=\cos$, respectively, rewritten as equivalent Prob-solvable loops.
The program in Listing~\ref{lst:protoex2-trig} has linear circular variable dependencies due to the variables $c$ and $s$.

\begin{minipage}{0.54\textwidth}
\begin{lstlisting}[mathescape,label={lst:protoex2-trig}, caption={PP prototype equivalent -- cos}]
$\displaystyle z = 0;\; x = 0;\; c = cos(z);\; s = sin(z);\; y = 0$
while true:
  $\displaystyle z = Normal(0,1)$
  $\displaystyle x = x + z$
  $\displaystyle c, \hspace{0.2cm} s = c \cdot \cos(z) - s \cdot \sin(z), \hspace{0.2cm} s \cdot \cos(z) + c \cdot \sin(z)$
  $\displaystyle y = y + c$ 
end
\end{lstlisting}
\end{minipage}
\begin{minipage}{0.44\textwidth}
\begin{lstlisting}[mathescape,label={lst:protoex2-exp}, caption={PP prototype equivalent -- exp}]
$\displaystyle z = 0;\; x = 0;\; \hat x_i = exp(z);\; y = 0$
while true:
  $\displaystyle z = Normal(0,1)$
  $\displaystyle x = x + z$
  $\displaystyle \hat x_i = \hat x_i \cdot \exp(z)$
  $\displaystyle y = y + \hat x_i$ 
end
\end{lstlisting}
\end{minipage}

\end{example}

\section{Evaluation}\label{sec:evaluation}

We evaluate our PCE-based method for moment approximation and our exact moment derivation approach on eleven benchmarks.
The set of benchmarks consists of those in \cite{KofnovMSBB22},  five additional benchmarks from \cite{Jasouretal2021}, and a probabilistic loop modeling stochastic exponential decay.
The benchmark \emph{Walking Robot} in \cite{Jasouretal2021} is the same as the \emph{Rimless wheel walker} in \cite{KofnovMSBB22}.

We apply our PCE-based method to approximate non-polynomial functions. This transforms all benchmark programs into Prob-solvable loops, which allows using the static analysis tool \textsc{Polar}~\cite{Moosbruggeretal2022} to compute the moments of the program variables as a function of the loop iteration $n$.

We implemented the techniques for exact moment derivation for loops containing trigonometric or exponential polynomials, presented in Section~\ref{sec:exact}, in the tool \textsc{Polar}.
We evaluate the technique for exact moment derivation using \textsc{Polar} on all benchmarks satisfying the general program structure of Listing \ref{lst:protoex} in Section~\ref{sec:exact}.
We also compare our approximate and exact methods with the technique based on \emph{polynomial forms} of \cite{Srirametal2020}. When appropriate, we applied our methods, as well as the polynomial form, on the eleven benchmark models. 
All experiments were run on a machine with 32 GB of RAM and a 2.6 GHz Intel i7 (Gen 10) processor.

\begin{figure}[!t]
\centering
  \includegraphics[scale=.30]{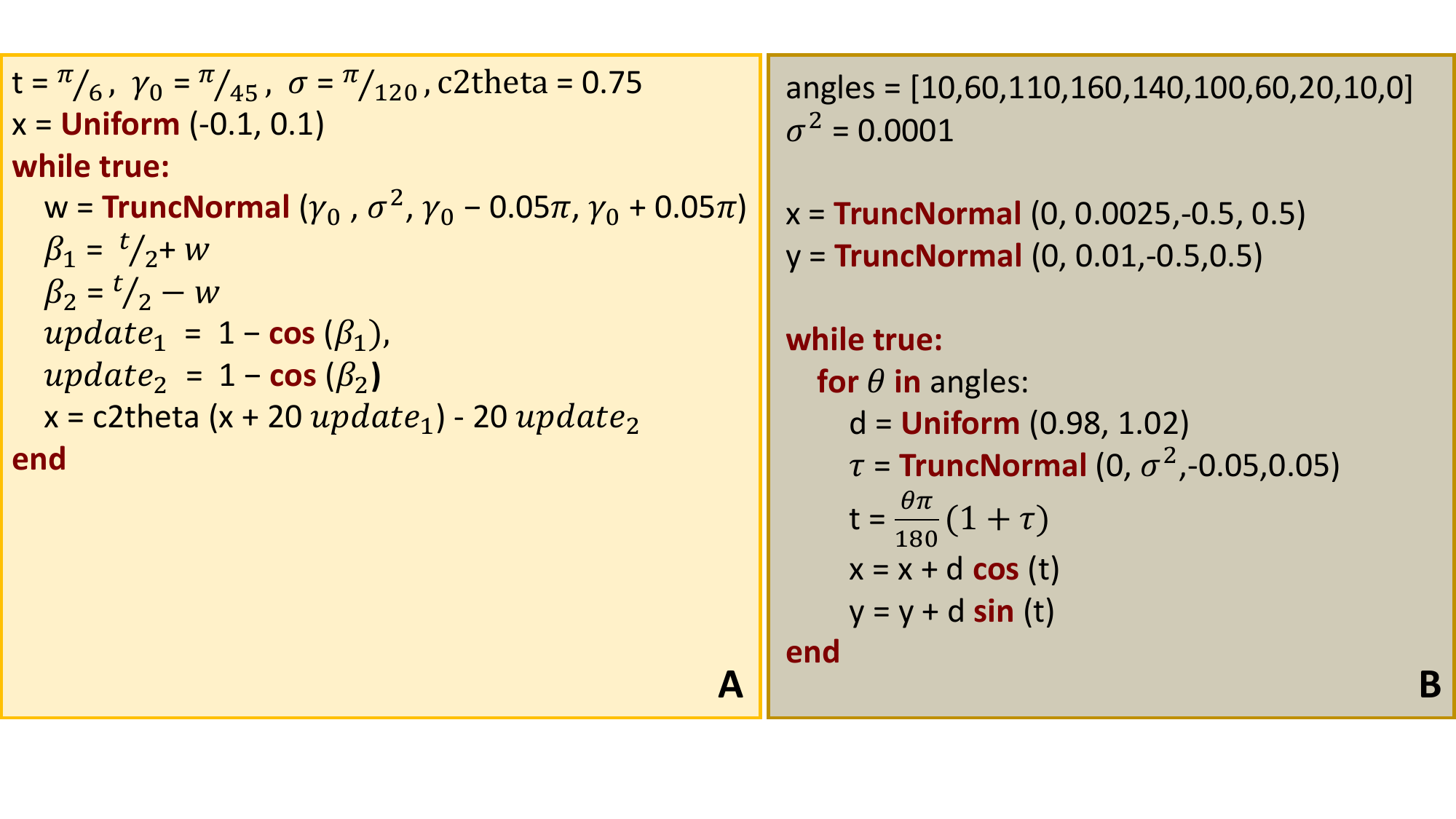}
  \caption{Probabilistic loops: (A) Rimless wheel walker~\cite{SteinhardtT12} and (B) 2D Robotic arm~\cite{Bouissouetal2016} (in the figure we use the inner loop as syntax sugar to keep the program compact). } 
  \label{fig:examplesA}
\end{figure}

\begin{figure}[!t]
\centering
  \includegraphics[scale=.450]{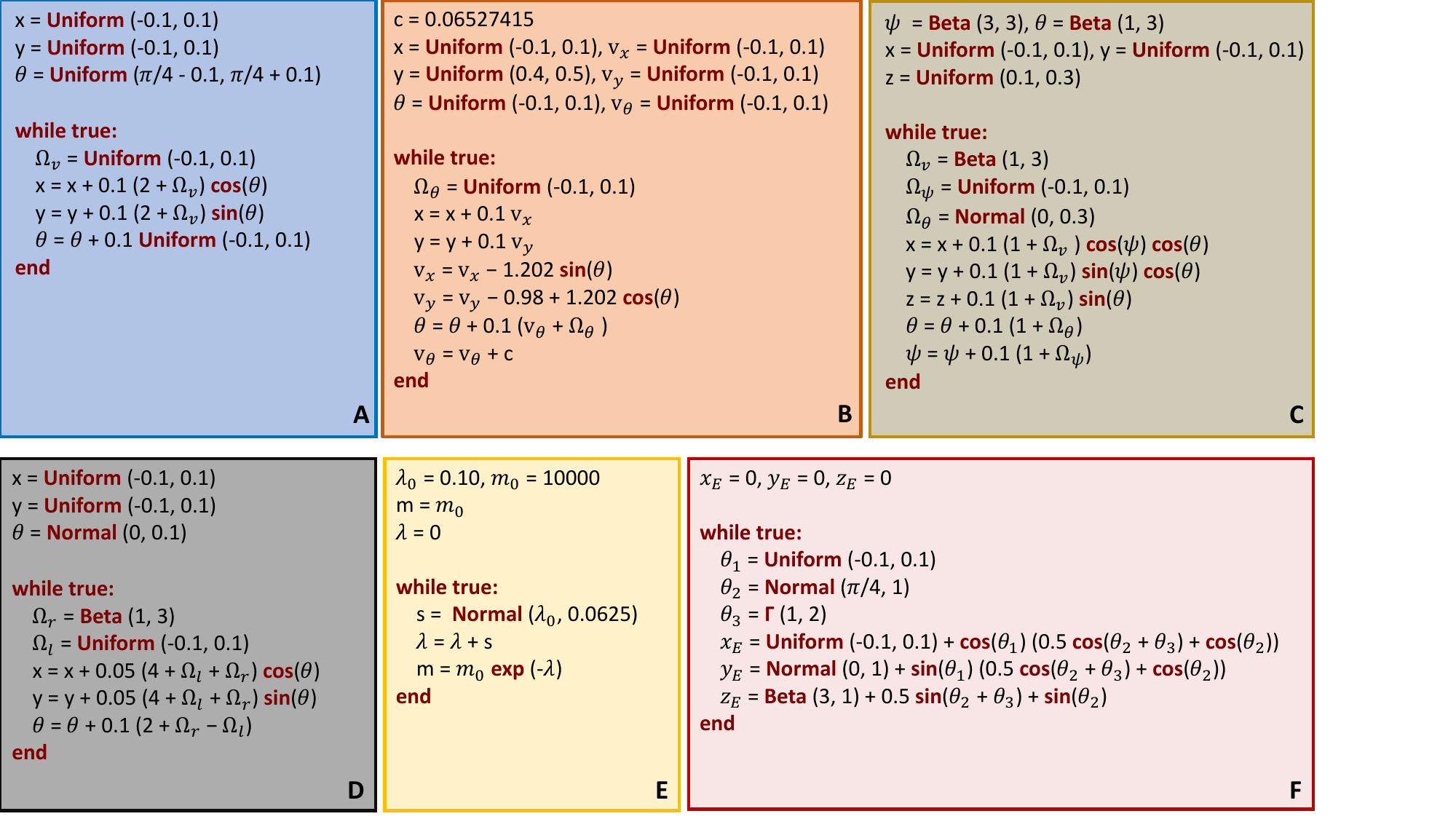}
  \caption{Probabilistic loops: (A) Uncertain underwater vehicle~\cite{Pairetetal_2020,Jasouretal2021}, (B)~Planar aerial vehicle~\cite{SteinhardtT12,Jasouretal2021}, (C)~3D aerial vehicle~\cite{Pairetetal_2020,Jasouretal2021}, (D)~Differential-drive mobile robot~\cite{vandenBergetal_2011,Jasouretal2021}, (E)~Stochastic decay, (F)~3D (Mobile) Robotic arm~\cite{Jasouretal2014,Jasouretal2021} } 
  \label{fig:examplesB}
\end{figure}

\bigskip

{\bf Taylor rule model.}
Central banks set monetary policy by raising or lowering their target for the federal funds rate.
The Taylor rule\footnote{It was proposed by the American economist John B. Taylor as a technique to stabilize economic activity by setting an interest rate \cite{Taylor1993}.} is an equation intended to describe the interest rate decisions of central banks.
The  rule relates the target of the federal funds rate to the current state of the economy through the formula
\begin{equation}
    i_{t} = r^{*}_{t} + \pi_{t} + a_{\pi}(\pi_{t} - \pi^{*}_{t}) + a_{y}(y_{t} - \bar{y}_{t}),
\end{equation}
where $i_{t}$ is the nominal interest rate, $r^{*}_{t}$ is the equilibrium real interest rate, $r^{*}_{t} = r$, $\pi_{t}$ is inflation rate at  $t$, $\pi^{*}_{t}$ is the short-term target inflation rate at $t$, $y_{t} = \log(1 + Y_{t})$, with $Y_{t}$  the real GDP, and $\bar{y}_{t} = \log(1 + \bar{Y}_{t})$, with $\bar{Y}_{t}$ denoting the potential real output. 

Highly-developed economies grow exponentially with a sufficiently small rate (e.g., according to the World Bank,\footnote{\url{https://data.worldbank.org/indicator/NY.GDP.MKTP.KD.ZG?locations=US}} the average growth rate of the GDP in the USA in 2001-2020 equaled  1.73\%).
Accordingly, we set the growth rate of the potential output to 2\%.
We model inflation as a martingale process; that is, $\E_{t}\left[\pi_{t+1}\right] = \pi_{t}$, following \cite{Atkeson_Ohanian_2001}.
The Taylor rule model is described by the program in Fig.~\ref{fig:taylor_model}.

\begin{table}[!t]
\resizebox{5.0in}{0.8in}{
        \centering
        \begin{tabular}{ccc}
            \includegraphics[scale=.8]{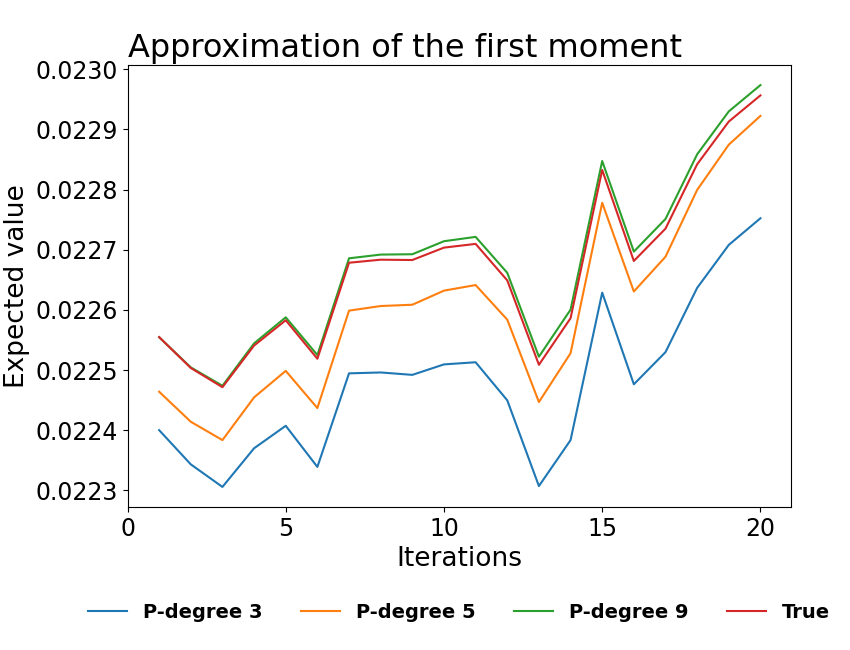} & \includegraphics[scale=.8]{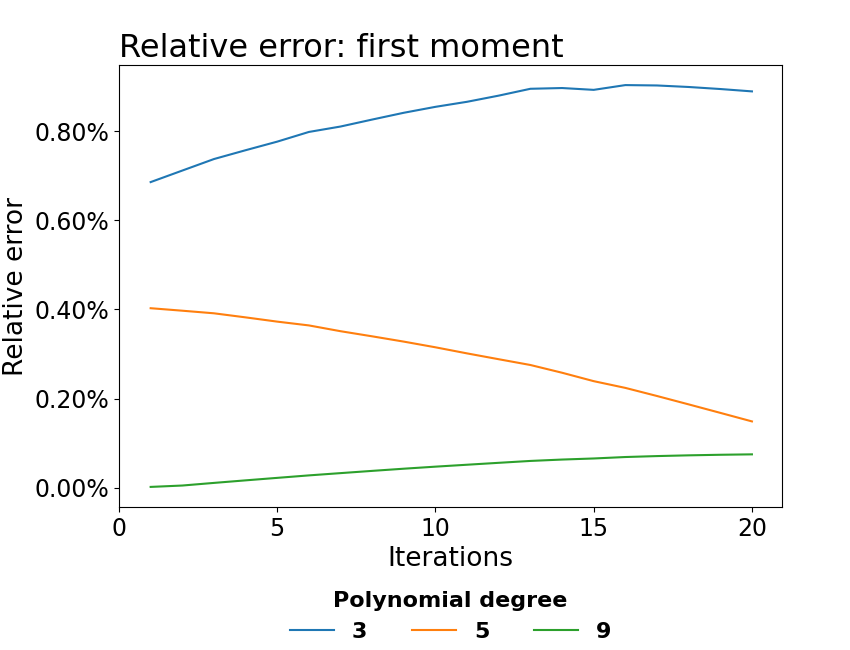} & 
            \includegraphics[scale=.8]{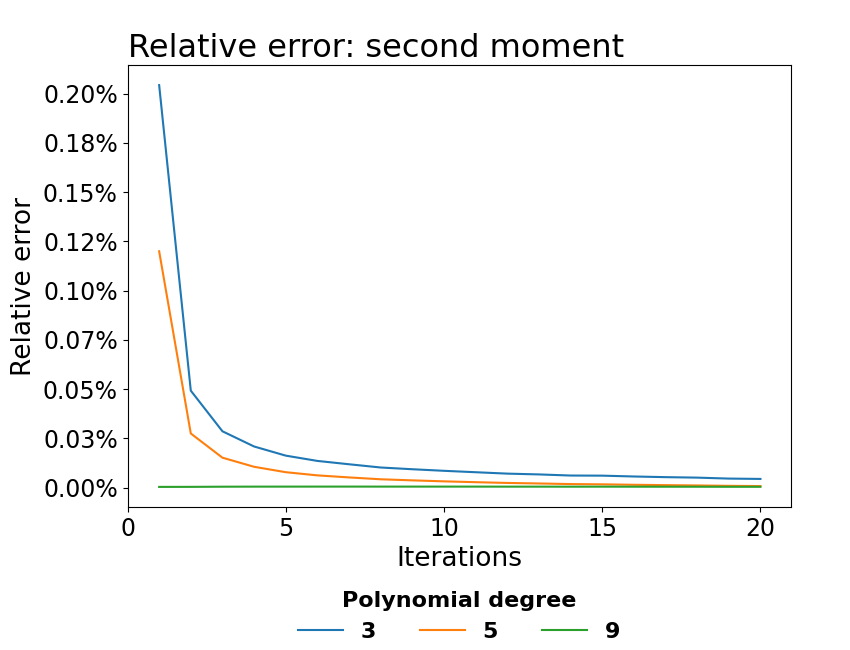} 
        \end{tabular}
}
        
        \captionof{figure}{The approximations and their relative errors for the Taylor rule model.} 
        \label{fig:aside}

    \end{table}

\medskip

{\bf Turning vehicle model.}
This model is described by the probabilistic program in Fig.~\ref{fig:turningexample}.
It was introduced in \cite{Srirametal2020} and depicts the position of a vehicle, as follows. 
The state variables are $(x, y, v, \psi)$, where $(x, y)$ is the vehicle's position with velocity $v$ and yaw angle $\psi$.
The vehicle’s velocity is stabilized around $v_{0} = 10$ m/s.
The dynamics are modelled by the equations  $x(t{+}1) = x(t) + \tau  v  \cos(\psi(t))$, $y(t{+}1) = y(t) + \tau v \sin(\psi(t))$, $v(t{+}1) = v(t) + \tau (K(v(t) - v_{0}) + w_{1}(t{+}1))$, and $\psi(t{+}1) = \psi(t) + w_{2}(t{+}1)$.
The disturbances $w_{1}$ and $w_{2}$ have distributions $w_{1} \sim U[-0.1, 0.1]$, $w_{2} \sim N(0, 0.1)$.
We set $K = -0.5$, as in \cite{Srirametal2020}.
Initially, the state variables are distributed as follows:
$x(0) \sim U[-0.1, 0.1]$, $y(0) \sim U[-0.5, -0.3]$, $v(0) \sim U[6.5, 8.0]$, $\psi(0) \sim N(0, 0.01)$.
We allow all normally distributed parameters to take values over the entire real line, in contrast to \cite{Srirametal2020} which could not accommodate distributions with infinite support and required the normal variables to be truncated.

\medskip

{\bf Rimless wheel walker.}
The \emph{Rimless wheel walker}~\cite{Srirametal2020,SteinhardtT12} is a system that describes a walking human.
The system models a rotating wheel consisting of $n_s$ spokes, each of length $L$, connected at a single point.
The angle between consecutive spokes is $\theta=2\pi / n_s$.
We set $L = 1$ and $\theta = \pi/6.$
This system  is modeled by the program in Fig.~\ref{fig:examplesA} (A).
For more details, we refer to~\cite{Srirametal2020}.

\medskip

{\bf Robotic arm model.}
Proposed and studied in \cite{Bouissouetal2016,Sankaranarayanan2020,Srirametal2020}, this system models the position of a 2D robotic arm.
The arm  moves through translations and rotations.
At every step, errors in movement are modeled with probabilistic noise.
The \textit{robotic arm model} is described by the program in Fig.~\ref{fig:examplesA}~(B).

\medskip

{\bf Uncertain underwater vehicle.}
This benchmark models the movement of an underwater vehicle subject to external disturbances~\cite{Jasouretal2021,Pairetetal_2020} and is encoded by the program in Fig.~\ref{fig:examplesB} (A).
The program variables $x$ and $y$ represent the position of the vehicle in a 2D plane and $\theta$ its orientation.
The external disturbances are modeled by probabilistic shocks to the velocity and the orientation of the vehicle.

\medskip

{\bf Planar aerial vehicle.}
This benchmark was studied in \cite{Jasouretal2021,SteinhardtT12} and models the vertical and horizontal movement of an aerial vehicle subject to wind disturbances.
It can be written as the program  in Fig.~\ref{fig:examplesB} (B), where the variables $x$ and $y$ represent the horizontal and vertical positions.
The variable $\theta$ models the rotation around the $x$ axis.
The linear velocities are captured by $v_x$ and $v_y$, and  $v_\theta$ represents the angular velocity.
The wind disturbance is modeled by the random variable~$\Omega_\theta$.
For more details, we refer to \cite{Jasouretal2021}.

\medskip

{\bf 3D aerial vehicle.}
This system, studied in \cite{Jasouretal2021,Pairetetal_2020}, models the movement of an aerial vehicle in three-dimensional space subject to wind disturbances.
The system can be written as a program as illustrated in Fig.~\ref{fig:examplesB} (C).
The program variables $x$, $y$, and $z$ represent the position of the vehicle.
The orientations around the $y$ and $z$ axis are captured by the variables $\theta$ and $\psi$, respectively.
The linear and angular velocities are constant $1$.
Wind disturbances are modeled by the random variables $\Omega_\nu$, $\Omega_\psi$, and $\Omega_\theta$.
For more details, we refer to \cite{Jasouretal2021}.

\medskip

{\bf Differential-drive mobile robot.}
This system models the movement of a differential-drive mobile robot with two wheels subject to external disturbances and was studied in \cite{Jasouretal2021,vandenBergetal_2011}.
In Fig.~\ref{fig:examplesB}~(D) we express the \emph{Differential-drive mobile robot} system as a program.
The program variables $x$ and~$y$ represent the robot's position.
Its orientation is captured by the variable $\theta$.
The velocities are constant $1$ for the left wheel and constant $3$ for the right wheel.
The random variables $\Omega_r$ and~$\Omega_l$ model external disturbances.
For more details, we refer to \cite{Jasouretal2021}.

\medskip

{\bf Mobile robotic arm.}
The system, studied in \cite{Jasouretal2021,Jasouretal2014,JasourF10}, models the uncertain position of the end-effector of a mobile robotic arm as a function of the uncertain base position and uncertain joint angles.
Fig.~\ref{fig:examplesB} (F) shows the system as a program.
The program variables $\theta_1$, $\theta_2$, and $\theta_3$ represent the uncertain angles of three joints.
The distributions of $\theta_1$ and $\theta_2$ are uniform and normal, respectively, while~$\theta_3$ is gamma distributed with shape parameter $1$ and scale parameter $2$.
The position of the end-effector in 3D space is given by the variables $x_E$, $y_E$, and $z_E$.
The uncertain position of the base in 3D space is modeled by three different distributions (uniform, normal, beta) in the assignments of $x_E$, $y_E$, and $z_E$.
For more details, we refer to \cite{Jasouretal2021}.

\medskip

{\bf Stochastic decay.}
The program in Fig.~\ref{fig:examplesB} (E) models  exponential decay with a non-constant stochastic decay rate.
Variable $m$ represents the total quantity subject to decay, where $m_0$ is the initial quantity.
The decay rate $\lambda$ starts off at $0$ and changes according to a normal distribution at every time step.

\begin{table}[t!]
\centering
\setlength\belowcaptionskip{8pt}
\footnotesize
\renewcommand{\arraystretch}{1.1}
\setlength\abovecaptionskip{8pt}
\setlength{\tabcolsep}{3pt}

\begin{tabular}{@{}lccccccccc@{}}
\rotatebox{0}{\textbf{Benchmark}} & \rotatebox{0}{\textbf{Target}} & \rotatebox{0}{\textbf{Poly form}} & \rotatebox{0}{\textbf{Sim.}}  & \rotatebox{0}{\textbf{Exact}} &  \multicolumn{3}{c}{\textbf{PCE estimate}}\\ 
 &  &  &  &   & \rotatebox{0}{\textbf{Deg.}} & \rotatebox{0}{\textbf{Result}} & \rotatebox{0}{\showclock{0}{0} \textbf{Runtime}} \\
\toprule

\begin{tabular}{@{}l@{}} Taylor rule \\ model \end{tabular} &
\begin{tabular}{@{}c@{}} $\E \left(i_n\right)$ \\ $n{=}20$ \end{tabular} & 
$\times$ &
0.022998 &
$\times$

&
\begin{tabular}{@{}c@{}} $3$ \\ $5$ \\ $9$ \end{tabular} & 
\begin{tabular}{@{}c@{}} $0.02278$ \\ $0.02295$ \\ $0.02300$ \end{tabular} &
\begin{tabular}{@{}c@{}}  $0.4s{+}0.5s$ \\  $0.5s{+}5.0s$ \\  $5.9s{+}34.6s$ \end{tabular} \\ 
\midrule

\begin{tabular}{@{}l@{}} Turning vehicle \\ model \end{tabular} &
\begin{tabular}{@{}c@{}} $\E \left(x_n\right)$ \\ $n{=}20$ \end{tabular} &
$\times$ &
15.60666 &
\begin{tabular}{@{}c@{}} \\ 15.60760 \\ \showclock{0}{37} 1.9s \end{tabular} &
\begin{tabular}{@{}c@{}} $3$ \\ $5$ \\ $9$ \end{tabular} &
\begin{tabular}{@{}c@{}} $14.44342$ \\ $15.43985$ \\ $15.60595$ \end{tabular} &
\begin{tabular}{@{}c@{}} $0.6s{+}3.6s$ \\ $1.4s{+}9.2s$ \\  $15.6s{+}16.1s$ \end{tabular} \\ 
\midrule

\begin{tabular}{@{}l@{}} Turning vehicle \\ model (trunc.) \end{tabular} &
\begin{tabular}{@{}c@{}} $\E \left(x_n\right)$ \\ $n{=}20$ \end{tabular} &
\begin{tabular}{@{}c@{}} for deg. 2 \\ $\left[-3 \cdot 10^5, 3 \cdot 10^5\right]$ \\ \showclock{0}{37} 1057s \end{tabular} &
15.60818 &
\begin{tabular}{@{}c@{}} \\ 15.60760 \\ \showclock{0}{37} 89.2s \end{tabular} &
\begin{tabular}{@{}c@{}} $3$ \\ $5$ \\ $9$ \end{tabular} &
\begin{tabular}{@{}c@{}} $14.44342$ \\ $15.43985$ \\ $15.60595$ \end{tabular} &
\begin{tabular}{@{}c@{}} $0.6s{+}3.6s$ \\ $1.4s{+}9.1s$ \\  $15.6s{+}15.8s$ \end{tabular} \\
\midrule

\begin{tabular}{@{}l@{}} Rimless wheel \\ walker \end{tabular} &
\begin{tabular}{@{}c@{}} $\E \left(x_n\right)$ \\ $n{=}2000$ \end{tabular} &
\begin{tabular}{@{}c@{}} for deg. 2 \\ $\left[1.791,1.792\right]$ \\ \showclock{0}{46} 5.42s \end{tabular} &
1.79173 &
\begin{tabular}{@{}c@{}} \\ 1.79159 \\ \showclock{0}{37} 8.0s \end{tabular} &
\begin{tabular}{@{}c@{}} $1$ \\ $2$ \\ $3$ \end{tabular} &
\begin{tabular}{@{}c@{}} $1.79159$ \\ $1.79159$ \\ $1.79159$ \end{tabular} &
\begin{tabular}{@{}c@{}}  $0.2s{+}0.5s$ \\ $0.3s{+}0.4s$ \\  $0.6s{+}0.6s$ \end{tabular} \\ 
\midrule

\begin{tabular}{@{}l@{}} Robotic arm \\ model \end{tabular} &
\begin{tabular}{@{}c@{}} $\E \left(x_n\right)$ \\ $n{=}100$ \end{tabular} &
\begin{tabular}{@{}c@{}} for deg. 2 \\ $\left[268.87,268.88\right]$ \\ \showclock{0}{36} 2.74s \end{tabular} &
268.852 &
\begin{tabular}{@{}c@{}} \\ 268.85236 \\ \showclock{0}{37} 5.6s \end{tabular} &
\begin{tabular}{@{}c@{}} $1$ \\ $2$ \\ $3$ \end{tabular} &
\begin{tabular}{@{}c@{}} $268.85236$ \\ $268.85236$ \\ $268.85236$ \end{tabular} &
\begin{tabular}{@{}c@{}}  $1.3s{+}0.3s$ \\  $2.5s{+}0.6s$ \\  $4.8s{+}0.7s$ \end{tabular} \\

\midrule

\rowcolor{tbl-row-color}
\begin{tabular}{@{}l@{}} Uncertain \\ underwater vehicle \end{tabular} &
\begin{tabular}{@{}c@{}} $\E \left(x^{2}_n\right)$ \\ $n{=}10$ \end{tabular}
&
\begin{tabular}{@{}c@{}} for deg. 2 \\ $\left[1.9817,2.0252\right]$ \\ \showclock{0}{36} 2.9s \end{tabular} &
\begin{tabular}{@{}c@{}} 2.00332
\end{tabular}
&
\begin{tabular}{@{}c@{}} \\ 2.00339 \\ \showclock{0}{37} 0.6s \end{tabular}
&
\begin{tabular}{@{}c@{}} $3$ \\ $5$ \\ $8$ \end{tabular} &
\begin{tabular}{@{}c@{}} $2.08986$ \\ $2.04514$ \\ $2.00432$ \end{tabular} &
\begin{tabular}{@{}c@{}}  $0.1s{+}0.9s$ \\ $0.1s{+}2.8s$ \\  $0.6s{+}8.6s$ \end{tabular} \\ 
\midrule

\rowcolor{tbl-row-color}
\begin{tabular}{@{}l@{}} Planar aerial \\ vehicle \end{tabular} &
\begin{tabular}{@{}c@{}} $\E \left(y_n\right)$ \\ $n{=}10$ \end{tabular} &
\begin{tabular}{@{}c@{}} for deg. 2 \\ $\left[1.4306,1.4315\right]$ \\ \showclock{0}{36} 4.1s \end{tabular} &
1.43111 &
$\times$
&
\begin{tabular}{@{}c@{}} $6$ \\ $8$ \\ $10$ \end{tabular} &
\begin{tabular}{@{}c@{}} $1.42184$ \\ $1.43016$ \\ $1.43099$ \end{tabular} &
\begin{tabular}{@{}c@{}}  $0.2s{+}5.9s$ \\ $0.6s{+}13.7s$ \\  $2.1s{+}28.0s$ \end{tabular} \\ 
\midrule

\rowcolor{tbl-row-color}
\begin{tabular}{@{}l@{}} 3D
aerial \\ vehicle \end{tabular} &
\begin{tabular}{@{}c@{}} $\E \left(x_n\right)$ \\ $n{=}20$ \end{tabular} &
$\times$ &
0.67736 &
\begin{tabular}{@{}c@{}} \\ 0.67770 \\ \showclock{0}{37} 4.9s \end{tabular} &
\begin{tabular}{@{}c@{}} $3$ \\ $5$ \\ $8$ \end{tabular} &
\begin{tabular}{@{}c@{}} $0.47805$ \\ $0.65280$ \\ $0.67245$ \end{tabular} &
\begin{tabular}{@{}c@{}}  $0.1s{+}1.5s$ \\ $0.1s{+}5.7s$ \\  $0.6s{+}30.5s$ \end{tabular} \\ 
\midrule

\rowcolor{tbl-row-color}
\begin{tabular}{@{}l@{}} Differential-drive  \\ mobile robot \end{tabular} &
\begin{tabular}{@{}c@{}} $\E \left(x_n^{2}\right)$ \\ $n{=}25$ \end{tabular} &
$\times$ &
0.29175 &
\begin{tabular}{@{}c@{}} \\ 0.29151 \\ \showclock{0}{37} 12.0s \end{tabular} &
\begin{tabular}{@{}c@{}} $8$ \\ $10$ \\ $12$ \end{tabular} &
\begin{tabular}{@{}c@{}} $0.19919$ \\ $0.29310$ \\ $0.29215$ \end{tabular} &
\begin{tabular}{@{}c@{}}  $0.6s{+}9.5s$ \\ $2.1s{+}13.8s$ \\  $8.3s{+}22.4s$ \end{tabular} \\ 
\midrule

\rowcolor{tbl-row-color}
\begin{tabular}{@{}l@{}} Mobile Robotic
 \\ Arm  \end{tabular} &
\begin{tabular}{@{}c@{}} $\E \left(x_n\right)$ \\ $n{=}2000$ \end{tabular} &
$\times$ &
0.38413 &
\begin{tabular}{@{}c@{}} \\ 0.38535 \\ \showclock{0}{37} 0.2s \end{tabular} &
\begin{tabular}{@{}c@{}} $2$ \\ $3$ \\ $4$ \end{tabular} &
\begin{tabular}{@{}c@{}} $0.38535$ \\ $0.38535$ \\ $0.38535$ \end{tabular} &
\begin{tabular}{@{}c@{}}  $0.8s{+}0.2s$ \\ $1.3s{+}0.3s$ \\  $2.0s{+}0.5s$ \end{tabular} \\
\midrule

\rowcolor{tbl-row-color}
\begin{tabular}{@{}l@{}} Stochastic decay  \end{tabular} &
\begin{tabular}{@{}c@{}} $\E \left(m_n\right)$ \\ $n{=}10$ \end{tabular} &
$\times$ &
5031.8404 &
\begin{tabular}{@{}c@{}} \\ 5028.3158 \\ \showclock{0}{37} 0.3s \end{tabular} &
\begin{tabular}{@{}c@{}} $6$ \\ $8$ \\ $10$ \end{tabular} &
\begin{tabular}{@{}c@{}} $ 5035.7468$ \\ $5028.0312$ \\ $ 5028.3222$ \end{tabular} &
\begin{tabular}{@{}c@{}}  $1.9s{+}1.0s$ \\ $4.7s{+}1.6s$ \\  $15.6s{+}2.0s$ \end{tabular} \\

\bottomrule
\end{tabular}
\caption{
Evaluation of our approach on $11$ benchmarks.
Poly form = the interval for the target as reported by \cite{Srirametal2020};
Sim = target approximated through $10^6$ samples;
Exact = the target result computed by our technique for exact moment derivation;
Deg. = maximum degrees used for the approximation of the non-linear functions;
Result = result of our approximate method per degree;
Runtime = execution time of our method in seconds (time of PCE + time of \textsc{Polar});
$\times$ = the respective method is not applicable.
The benchmarks in grey are new relative to \cite{KofnovMSBB22}.
}
\label{tab:LiOModl}
\end{table}

\medskip

Fig. \ref{fig:aside} illustrates the performance of our PCE-based approach as a function of the polynomial degree of our approximation on the \textit{Taylor rule}. The approximations to the true first moment (in red) are plotted in the left panel and the relative errors, calculated as $rel.err = |est - true|/true$, for the first and second moments  in the middle and  right panels, respectively, over iteration number. 
All plots show that the approximation error is low and deteriorates as the polynomial degree increases from 3 to 9, across iterations. For this benchmark, the drop is sharper for the second moment.

The \emph{Rimless wheel walker} and the \emph{Robotic arm} models are the only two benchmarks from \cite{Srirametal2020} with nonlinear non-polynomial updates.
Polynomial forms of degree $2$ were used to compute bounding intervals for $\E(x_n)$ (for fixed $n$) for these two models.
The \cite{Srirametal2020} tool  supports neither the approximation of logarithms (required for the \emph{Taylor rule model}) nor distributions with unbounded support (required for all benchmarks except for the \emph{Taylor rule model} on which the tool fails).
To facilitate comparison with polynomial forms, our set of benchmarks is augmented with a version of the \emph{Turning vehicle model} using truncated normal distributions ($\psi$ and $w_2 \sim TruncNormal(0, 0.01, \left[-1, 1\right])$ in Fig. \ref{fig:turningexample}), which is called \emph{Turning vehicle model (trunc.)}  in Table~\ref{tab:LiOModl}), instead of normal distributions with unbounded support.

Among the eleven benchmark models in Table \ref{tab:LiOModl}, the polynomial form tool of \cite{Srirametal2020} can be used to approximate moments only in five, namely the \emph{Turning vehicle model (trunc.)},  \emph{Rimless wheel walker},  \emph{Robotic arm}, \emph{Uncertain underwater vehicle}, and \emph{Planar aerial vehicle}. 
Our method for exact moment derivation supports trigonometric functions and the exponential function but no logarithms. Hence, it is not applicable to the \emph{Taylor rule model}.
Moreover, our exact method cannot be applied to the \emph{Planar aerial vehicle} benchmark because the perturbation of its program variable $\theta$ is not iteration-stable and $\theta$ is used as an argument to a trigonometric function. Our PCE-based moment estimation approach applies to all.

The \emph{Robotic arm},  \emph{Rimless wheel walker}, and \emph{Mobile robotic arm} models contain no stochastic accumulation: each basic random variable is iteration-stable and can be estimated using the scheme in Section~\ref{sec:stable}.
Therefore, for these benchmarks, our estimates converge exponentially fast to the true values.
In fact, our estimates coincide with the true values for first moments, because the estimators are unbiased.
The other benchmarks contain stochasticity accumulation, which leads to the instability of the distributions of basic random variables.
For these benchmarks, we apply the scheme in Section~\ref{sec:non-stable}.

Table~\ref{tab:LiOModl} contains the evaluation results of our approximate and exact approaches, and of the technique based on \textit{polynomial forms} of \cite{Srirametal2020} on the eleven benchmarks. In consecutive order, the table columns are: the name of the benchmark model; the target moment and iteration; the \textit{polynomial form} results (estimation interval and runtime), if applicable; the sampling-based value of the target moment;  the exact moment and the runtime of its calculation, if applicable;  the truncation parameter (polynomial degree) in PCE;  the PCE estimate value;  and the PCE estimate calculation runtime. 

Our results illustrate that our method based on PCE is able to accurately approximate general non-linear dynamics for challenging programs.
Specifically, for the \emph{Rimless wheel walker} model, our first moment estimate coincides with the exact result and falls in the interval estimate of the polynomial forms technique.
For the \emph{Robotic arm model}, our results are equal to the exact result and closer to the sampling  one based on $10^6$ samples. They lie outside the interval predicted by the polynomial forms technique, pointing to the latter's lack of accuracy in this model.

Our method for exact moment derivation can be faster than the polynomial form technique and our PCE-based approximation approach, for instance, for the \emph{Turning vehicle model}.
Nevertheless, if all basic random variables are iteration-stable, our approximation approach will provide an unbiased estimation and hence the exact result for the first moments.
This is the case, for example, for the \emph{Rimless wheel walker} benchmark for which our approximation method provides the true result in under $0.7$s, compared to our exact moment derivation method which needs $8$s.

Our experiments also demonstrate that our PCE-based method provides accurate approximations in a fraction of the time required by the polynomial form based technique.
While polynomial forms compute an error interval, they need to be computed on an iteration-by-iteration basis.
In contrast, our method based on PCE and Prob-solvable loops computes an expression for the target parameterized by the loop iteration $n \in \mathbb{N}$ (cf. Fig.~\ref{fig:turningexample}).
As a result, increasing the target iteration~$n$ does \emph{not} increase the runtime of our approach.
To see this,  consider the \emph{Uncertain underwater vehicle} benchmark:
the runtimes of polynomial forms and of our approach using the PCE estimate of order~$5$ are comparable~($2.9$s).
However, increasing the target iteration~$n$ from~$10$ to~$20$ escalates the runtime of polynomial forms to~$237$s while the runtimes of both our approaches (approximate and exact) remain the same.

\section{Conclusion}\label{sec:conclusion} 
We present two methods, one exact and one approximate, to compute the state variable moments in closed-form in probabilistic loops with non-polynomial updates.
Our approximation method is based on polynomial chaos expansion to approximate non-polynomial general functional assignments.
The approximations produced by our technique have optimal exponential convergence when the parameters of the general non-polynomial functions have distributions that are stable across all iterations.
We derive an upper bound on the approximation error for the case of unstable parameter distributions.
Our exact method is applicable to probabilistic loops with trigonometric and exponential assignments if the random perturbations of the arguments of the non-linear functions are independent across iterations.

Our methods can accommodate non-linear, non-polynomial updates in classes of probabilistic loops amenable to automated moment computation, such as the class of Prob-solvable loops.
We  emphasize  that our PCE-based approximation is not limited to Prob-solvable loops and can be applied to approximate non-linear dynamics in more general probabilistic loops. 

Our experiments demonstrate the ability of our methods to characterize non-polynomial behavior in stochastic models from various domains via their moments, with high accuracy and in a fraction of the time required by other state-of-the-art tools.
In future work, we plan to investigate how to use these solutions to automatically compute stability properties (e.g. Lyapunov stability and asymptotic stability) in stochastic dynamical systems.


\begin{acks}
The research in this paper has been funded by the Vienna Science and Technology Fund (WWTF) [10.47379/ICT19018], the TU Wien Doctoral College (SecInt), the FWF research
projects LogiCS W1255-N23 and P 30690-N35, and the ERC Consolidator Grant ARTIST 101002685.
\end{acks}
  
\bibliographystyle{ACM-Reference-Format}
\bibliography{refs}


\begin{thebibliography}{45}


\ifx \showCODEN    \undefined \def \showCODEN     #1{\unskip}     \fi
\ifx \showDOI      \undefined \def \showDOI       #1{#1}\fi
\ifx \showISBNx    \undefined \def \showISBNx     #1{\unskip}     \fi
\ifx \showISBNxiii \undefined \def \showISBNxiii  #1{\unskip}     \fi
\ifx \showISSN     \undefined \def \showISSN      #1{\unskip}     \fi
\ifx \showLCCN     \undefined \def \showLCCN      #1{\unskip}     \fi
\ifx \shownote     \undefined \def \shownote      #1{#1}          \fi
\ifx \showarticletitle \undefined \def \showarticletitle #1{#1}   \fi
\ifx \showURL      \undefined \def \showURL       {\relax}        \fi
\providecommand\bibfield[2]{#2}
\providecommand\bibinfo[2]{#2}
\providecommand\natexlab[1]{#1}
\providecommand\showeprint[2][]{arXiv:#2}

\bibitem[Amrollahi et~al\mbox{.}(2022)]%
        {AmrollahiBKKMS22}
\bibfield{author}{\bibinfo{person}{Daneshvar Amrollahi}, \bibinfo{person}{Ezio
  Bartocci}, \bibinfo{person}{George Kenison}, \bibinfo{person}{Laura
  Kov{\'{a}}cs}, \bibinfo{person}{Marcel Moosbrugger}, {and}
  \bibinfo{person}{Miroslav Stankovic}.} \bibinfo{year}{2022}\natexlab{}.
\newblock \showarticletitle{Solving Invariant Generation for Unsolvable Loops}.
  In \bibinfo{booktitle}{\emph{Proc. of {SAS} 2022: the 29th International
  Symposium on Static Analysis}} \emph{(\bibinfo{series}{LNCS},
  Vol.~\bibinfo{volume}{13790})}. \bibinfo{publisher}{Springer},
  \bibinfo{pages}{19--43}.
\newblock
\urldef\tempurl%
\url{https://doi.org/10.1007/978-3-031-22308-2\_3}
\showDOI{\tempurl}


\bibitem[Atkeson and Ohanian(2001)]%
        {Atkeson_Ohanian_2001}
\bibfield{author}{\bibinfo{person}{Andrew Atkeson} {and}
  \bibinfo{person}{Lee~E. Ohanian}.} \bibinfo{year}{2001}\natexlab{}.
\newblock \showarticletitle{Are Phillips curves useful for forecasting
  inflation?}
\newblock \bibinfo{journal}{\emph{Quarterly Review}} \bibinfo{volume}{25},
  \bibinfo{number}{Win} (\bibinfo{year}{2001}), \bibinfo{pages}{2--11}.
\newblock
\urldef\tempurl%
\url{https://ideas.repec.org/a/fip/fedmqr/y2001iwinp2-11nv.25no.1.html}
\showURL{%
\tempurl}


\bibitem[Bartocci et~al\mbox{.}(2019)]%
        {Bartoccietal2019}
\bibfield{author}{\bibinfo{person}{Ezio Bartocci}, \bibinfo{person}{Laura
  Kov{\'{a}}cs}, {and} \bibinfo{person}{Miroslav Stankovic}.}
  \bibinfo{year}{2019}\natexlab{}.
\newblock \showarticletitle{Automatic Generation of Moment-Based Invariants for
  Prob-Solvable Loops}. In \bibinfo{booktitle}{\emph{Proc. of {ATVA} 2019: the
  17th International Symposium on Automated Technology for Verification and
  Analysis}} \emph{(\bibinfo{series}{LNCS}, Vol.~\bibinfo{volume}{11781})}.
  \bibinfo{publisher}{Springer}, \bibinfo{pages}{255--276}.
\newblock
\urldef\tempurl%
\url{https://doi.org/10.1007/978-3-030-31784-3\_15}
\showDOI{\tempurl}


\bibitem[Bartocci et~al\mbox{.}(2020)]%
        {BartocciKS20a}
\bibfield{author}{\bibinfo{person}{Ezio Bartocci}, \bibinfo{person}{Laura
  Kov{\'{a}}cs}, {and} \bibinfo{person}{Miroslav Stankovic}.}
  \bibinfo{year}{2020}\natexlab{}.
\newblock \showarticletitle{\textsc{Mora} - Automatic Generation of
  Moment-Based Invariants}. In \bibinfo{booktitle}{\emph{Proc. of {TACAS} 2020:
  the 26th International Conference on Tools and Algorithms}}
  \emph{(\bibinfo{series}{LNCS}, Vol.~\bibinfo{volume}{12078})}.
  \bibinfo{publisher}{Springer}, \bibinfo{pages}{492--498}.
\newblock
\urldef\tempurl%
\url{https://doi.org/10.1007/978-3-030-45190-5\_28}
\showDOI{\tempurl}


\bibitem[Bouissou et~al\mbox{.}(2016)]%
        {Bouissouetal2016}
\bibfield{author}{\bibinfo{person}{Olivier Bouissou}, \bibinfo{person}{Eric
  Goubault}, \bibinfo{person}{Sylvie Putot}, \bibinfo{person}{Aleksandar
  Chakarov}, {and} \bibinfo{person}{Sriram Sankaranarayanan}.}
  \bibinfo{year}{2016}\natexlab{}.
\newblock \showarticletitle{Uncertainty propagation using probabilistic affine
  forms and concentration of measure inequalities}. In
  \bibinfo{booktitle}{\emph{Proc. of {TACAS} 2016: the 22nd International
  Conference on Tools and Algorithms for the Construction and Analysis of
  Systems}} \emph{(\bibinfo{series}{LNCS}, Vol.~\bibinfo{volume}{9636})}.
  Springer, \bibinfo{publisher}{Springer}, \bibinfo{pages}{225--243}.
\newblock
\urldef\tempurl%
\url{https://doi.org/10.1007/978-3-662-49674-9\_13}
\showDOI{\tempurl}


\bibitem[Chen et~al\mbox{.}(2012)]%
        {Chenetal12}
\bibfield{author}{\bibinfo{person}{Xin Chen}, \bibinfo{person}{Erika
  {\'{A}}brah{\'{a}}m}, {and} \bibinfo{person}{Sriram Sankaranarayanan}.}
  \bibinfo{year}{2012}\natexlab{}.
\newblock \showarticletitle{Taylor Model Flowpipe Construction for Non-linear
  Hybrid Systems}. In \bibinfo{booktitle}{\emph{Proc. of {IEEE} {RTSS}: the
  33rd {IEEE} Real-Time Systems Symposium}}. \bibinfo{publisher}{{IEEE}},
  \bibinfo{pages}{183--192}.
\newblock
\urldef\tempurl%
\url{https://doi.org/10.1109/RTSS.2012.70}
\showDOI{\tempurl}


\bibitem[Chorin(1974)]%
        {Chorin1974}
\bibfield{author}{\bibinfo{person}{Alexandre~Joel Chorin}.}
  \bibinfo{year}{1974}\natexlab{}.
\newblock \showarticletitle{Gaussian fields and random flow}.
\newblock \bibinfo{journal}{\emph{Journal of Fluid Mechanics}}
  \bibinfo{volume}{63}, \bibinfo{number}{1} (\bibinfo{year}{1974}),
  \bibinfo{pages}{21–32}.
\newblock
\urldef\tempurl%
\url{https://doi.org/10.1017/S0022112074000991}
\showDOI{\tempurl}


\bibitem[Denamiel et~al\mbox{.}(2020)]%
        {Denamieletal2020}
\bibfield{author}{\bibinfo{person}{Cl\'ea Denamiel}, \bibinfo{person}{Xun
  Huan}, \bibinfo{person}{Jadranka Šepić}, {and} \bibinfo{person}{Ivica
  Vilibi\'c}.} \bibinfo{year}{2020}\natexlab{}.
\newblock \showarticletitle{Uncertainty Propagation Using Polynomial Chaos
  Expansions for Extreme Sea Level Hazard Assessment: The Case of the Eastern
  Adriatic Meteotsunamis}.
\newblock \bibinfo{journal}{\emph{Journal of Physical Oceanography}}
  \bibinfo{volume}{50}, \bibinfo{number}{4} (\bibinfo{year}{2020}),
  \bibinfo{pages}{1005 -- 1021}.
\newblock
\urldef\tempurl%
\url{https://doi.org/10.1175/JPO-D-19-0147.1}
\showDOI{\tempurl}


\bibitem[Durrett(2019)]%
        {Durrett2019}
\bibfield{author}{\bibinfo{person}{Rick Durrett}.}
  \bibinfo{year}{2019}\natexlab{}.
\newblock \bibinfo{booktitle}{\emph{Probability: Theory and Examples}}.
\newblock \bibinfo{publisher}{Cambridge University Press}.
\newblock
\urldef\tempurl%
\url{https://doi.org/10.1017/9781108591034}
\showDOI{\tempurl}


\bibitem[{Ernst, Oliver G.} et~al\mbox{.}(2012)]%
        {Ernstetal2012}
\bibfield{author}{\bibinfo{person}{{Ernst, Oliver G.}},
  \bibinfo{person}{{Mugler, Antje}}, \bibinfo{person}{{Starkloff,
  Hans-J\"org}}, {and} \bibinfo{person}{{Ullmann, Elisabeth}}.}
  \bibinfo{year}{2012}\natexlab{}.
\newblock \showarticletitle{On the convergence of generalized polynomial chaos
  expansions}.
\newblock \bibinfo{journal}{\emph{ESAIM: M2AN}} \bibinfo{volume}{46},
  \bibinfo{number}{2} (\bibinfo{year}{2012}), \bibinfo{pages}{317--339}.
\newblock
\urldef\tempurl%
\url{https://doi.org/10.1051/m2an/2011045}
\showDOI{\tempurl}


\bibitem[Foo et~al\mbox{.}(2007)]%
        {Fooetal2007}
\bibfield{author}{\bibinfo{person}{Jasmine Foo}, \bibinfo{person}{Zohar
  Yosibash}, {and} \bibinfo{person}{Georg~Em Karniadakis}.}
  \bibinfo{year}{2007}\natexlab{}.
\newblock \showarticletitle{Stochastic simulation of riser-sections with
  uncertain measured pressure loads and/or uncertain material properties}.
\newblock \bibinfo{journal}{\emph{Comput. Methods Appl. Mech. Eng.}}
  \bibinfo{volume}{196} (\bibinfo{year}{2007}), \bibinfo{pages}{4250--4271}.
\newblock
\urldef\tempurl%
\url{https://doi.org/10.1016/j.cma.2007.04.005}
\showDOI{\tempurl}


\bibitem[Formaggia et~al\mbox{.}(2013)]%
        {Formaggiaetal2013}
\bibfield{author}{\bibinfo{person}{Luca Formaggia}, \bibinfo{person}{Alberto
  Guadagnini}, \bibinfo{person}{Ilaria Imperiali}, \bibinfo{person}{Valentina
  Lever}, \bibinfo{person}{Giovanni Porta}, \bibinfo{person}{Monica Riva},
  \bibinfo{person}{Anna Scotti}, {and} \bibinfo{person}{Lorenzo Tamellini}.}
  \bibinfo{year}{2013}\natexlab{}.
\newblock \showarticletitle{Global sensitivity analysis through polynomial
  chaos expansion of a basin-scale geochemical compaction model}.
\newblock \bibinfo{journal}{\emph{Comput. Geosci.}}  \bibinfo{volume}{17}
  (\bibinfo{year}{2013}), \bibinfo{pages}{25--42}.
\newblock
\urldef\tempurl%
\url{https://doi.org/10.1007/s10596-012-9311-5}
\showDOI{\tempurl}


\bibitem[Ghanem(1998)]%
        {Ghanem1998}
\bibfield{author}{\bibinfo{person}{Roger Ghanem}.}
  \bibinfo{year}{1998}\natexlab{}.
\newblock \showarticletitle{Probabilistic Characterization of Transport in
  Heterogeneous Media}.
\newblock \bibinfo{journal}{\emph{Computer Methods in Applied Mechanics and
  Engineering}}  \bibinfo{volume}{158} (\bibinfo{year}{1998}),
  \bibinfo{pages}{199–220}.
\newblock
\urldef\tempurl%
\url{https://doi.org/10.1016/s0045-7825(97)00250-8}
\showDOI{\tempurl}


\bibitem[Ghanem and Dham(1998)]%
        {GhanemDham1998}
\bibfield{author}{\bibinfo{person}{R. Ghanem} {and} \bibinfo{person}{S. Dham}.}
  \bibinfo{year}{1998}\natexlab{}.
\newblock \showarticletitle{Stochastic Finite Element Analysis for Multiphase
  Flow in Heterogeneous Porous Media}.
\newblock \bibinfo{journal}{\emph{Transport in Porous Medias}}
  \bibinfo{volume}{32} (\bibinfo{year}{1998}), \bibinfo{pages}{239–262}.
\newblock
\urldef\tempurl%
\url{https://doi.org/10.1023/A:1006514109327}
\showURL{%
\tempurl}


\bibitem[Ghanem and Spanos(1991)]%
        {GhanemSpanos1991}
\bibfield{author}{\bibinfo{person}{Roger~G. Ghanem} {and}
  \bibinfo{person}{Pol~D. Spanos}.} \bibinfo{year}{1991}\natexlab{}.
\newblock \bibinfo{booktitle}{\emph{Stochastic Finite Elements: A Spectral
  Approach}}.
\newblock \bibinfo{publisher}{Springer}, \bibinfo{address}{New York, NY}.
\newblock


\bibitem[Giraldi et~al\mbox{.}(2017)]%
        {Giraldietal2017}
\bibfield{author}{\bibinfo{person}{Lo\"ic Giraldi}, \bibinfo{person}{Olivier
  P.~Le Ma{\^{\i}}tre}, \bibinfo{person}{Kyle~T. Mandli},
  \bibinfo{person}{Clint~N. Dawson}, \bibinfo{person}{Ibrahim Hoteit}, {and}
  \bibinfo{person}{Omar~M. Knio}.} \bibinfo{year}{2017}\natexlab{}.
\newblock \showarticletitle{Bayesian inference of earthquake parameters from
  buoy data using a polynomial chaos-based surrogate}.
\newblock \bibinfo{journal}{\emph{Comput. Geosci.}}  \bibinfo{volume}{21}
  (\bibinfo{year}{2017}), \bibinfo{pages}{683--699}.
\newblock
\urldef\tempurl%
\url{https://doi.org/10.1007/s10596-017-9646-z}
\showDOI{\tempurl}


\bibitem[Hien and Kleiber(1997)]%
        {HienKleiber1997}
\bibfield{author}{\bibinfo{person}{Tran~Duong Hien} {and}
  \bibinfo{person}{Michał Kleiber}.} \bibinfo{year}{1997}\natexlab{}.
\newblock \showarticletitle{Stochastic finite element modelling in linear
  transient heat transfer}.
\newblock \bibinfo{journal}{\emph{Computer Methods in Applied Mechanics and
  Engineering}} \bibinfo{volume}{144}, \bibinfo{number}{1}
  (\bibinfo{year}{1997}), \bibinfo{pages}{111--124}.
\newblock
\showISSN{0045-7825}
\urldef\tempurl%
\url{https://doi.org/10.1016/S0045-7825(96)01168-1}
\showDOI{\tempurl}


\bibitem[Hou et~al\mbox{.}(2006)]%
        {Houetal2006}
\bibfield{author}{\bibinfo{person}{Thomas~Y. Hou}, \bibinfo{person}{Wuan Luo},
  \bibinfo{person}{Boris Rozovskii}, {and} \bibinfo{person}{Hao-Min Zhou}.}
  \bibinfo{year}{2006}\natexlab{}.
\newblock \showarticletitle{Wiener chaos expansions and numerical solutions of
  randomly forced equations of fluid mechanics}.
\newblock \bibinfo{journal}{\emph{J. Comput. Phys.}}  \bibinfo{volume}{216}
  (\bibinfo{year}{2006}), \bibinfo{pages}{687--706}.
\newblock
\urldef\tempurl%
\url{https://doi.org/10.1016/j.jcp.2006.01.008}
\showDOI{\tempurl}


\bibitem[Jasour et~al\mbox{.}(2021)]%
        {Jasouretal2021}
\bibfield{author}{\bibinfo{person}{Ashkan Jasour}, \bibinfo{person}{Allen
  Wang}, {and} \bibinfo{person}{Brian~C. Williams}.}
  \bibinfo{year}{2021}\natexlab{}.
\newblock \bibinfo{title}{Moment-Based Exact Uncertainty Propagation Through
  Nonlinear Stochastic Autonomous Systems}.
\newblock
\newblock
\showeprint[arXiv]{2101.12490}
\urldef\tempurl%
\url{https://arxiv.org/abs/2101.12490}
\showURL{%
\tempurl}


\bibitem[Jasour and Farrokhi(2014)]%
        {Jasouretal2014}
\bibfield{author}{\bibinfo{person}{Ashkan~M. Jasour} {and}
  \bibinfo{person}{Mohammad Farrokhi}.} \bibinfo{year}{2014}\natexlab{}.
\newblock \showarticletitle{Adaptive neuro-predictive control for redundant
  robot manipulators in presence of static and dynamic obstacles: A
  Lyapunov-based approach}.
\newblock \bibinfo{journal}{\emph{International Journal of Adaptive Control and
  Signal Processing}} \bibinfo{volume}{28}, \bibinfo{number}{3-5}
  (\bibinfo{year}{2014}), \bibinfo{pages}{386--411}.
\newblock
\urldef\tempurl%
\url{https://doi.org/10.1002/acs.2459}
\showDOI{\tempurl}
\showeprint{https://onlinelibrary.wiley.com/doi/pdf/10.1002/acs.2459}


\bibitem[Jasour and Farrokhi(2010)]%
        {JasourF10}
\bibfield{author}{\bibinfo{person}{Ashkan M.~Z. Jasour} {and}
  \bibinfo{person}{Mohammad Farrokhi}.} \bibinfo{year}{2010}\natexlab{}.
\newblock \showarticletitle{Fuzzy improved adaptive neuro-NMPC for online path
  tracking and obstacle avoidance of redundant robotic manipulators}.
\newblock \bibinfo{journal}{\emph{Int. J. Autom. Control.}}
  \bibinfo{volume}{4}, \bibinfo{number}{2} (\bibinfo{year}{2010}),
  \bibinfo{pages}{177--200}.
\newblock
\urldef\tempurl%
\url{https://doi.org/10.1504/IJAAC.2010.030810}
\showDOI{\tempurl}


\bibitem[Karimi et~al\mbox{.}(2022)]%
        {KarimiMSKBB22}
\bibfield{author}{\bibinfo{person}{Ahmad Karimi}, \bibinfo{person}{Marcel
  Moosbrugger}, \bibinfo{person}{Miroslav Stankovic}, \bibinfo{person}{Laura
  Kov{\'{a}}cs}, \bibinfo{person}{Ezio Bartocci}, {and}
  \bibinfo{person}{Efstathia Bura}.} \bibinfo{year}{2022}\natexlab{}.
\newblock \showarticletitle{Distribution Estimation for Probabilistic Loops}.
  In \bibinfo{booktitle}{\emph{Proc. of {QEST} 2022: the 19th International
  Conference on Quantitative Evaluation of Systems}}
  \emph{(\bibinfo{series}{LNCS}, Vol.~\bibinfo{volume}{13479})}.
  \bibinfo{publisher}{Springer}, \bibinfo{pages}{26--42}.
\newblock
\urldef\tempurl%
\url{https://doi.org/10.1007/978-3-031-16336-4\_2}
\showDOI{\tempurl}


\bibitem[Kauers and Paule(2011)]%
        {KauersP11}
\bibfield{author}{\bibinfo{person}{Manuel Kauers} {and} \bibinfo{person}{Peter
  Paule}.} \bibinfo{year}{2011}\natexlab{}.
\newblock \bibinfo{booktitle}{\emph{The Concrete Tetrahedron - Symbolic Sums,
  Recurrence Equations, Generating Functions, Asymptotic Estimates}}.
\newblock \bibinfo{publisher}{Springer}.
\newblock
\urldef\tempurl%
\url{https://doi.org/10.1007/978-3-7091-0445-3}
\showDOI{\tempurl}


\bibitem[Knio and Ma{\^{\i}}tre(2006)]%
        {Knio2006}
\bibfield{author}{\bibinfo{person}{Omar~M Knio} {and}
  \bibinfo{person}{Oliver~Le Ma{\^{\i}}tre}.} \bibinfo{year}{2006}\natexlab{}.
\newblock \showarticletitle{Uncertainty propagation in {CFD} using polynomial
  chaos decomposition}.
\newblock \bibinfo{journal}{\emph{Fluid Dynamics Research}}
  \bibinfo{volume}{38}, \bibinfo{number}{9} (\bibinfo{date}{sep}
  \bibinfo{year}{2006}), \bibinfo{pages}{616--640}.
\newblock
\urldef\tempurl%
\url{https://doi.org/10.1016/j.fluiddyn.2005.12.003}
\showDOI{\tempurl}


\bibitem[Kofnov et~al\mbox{.}(2022)]%
        {KofnovMSBB22}
\bibfield{author}{\bibinfo{person}{Andrey Kofnov}, \bibinfo{person}{Marcel
  Moosbrugger}, \bibinfo{person}{Miroslav Stankovic}, \bibinfo{person}{Ezio
  Bartocci}, {and} \bibinfo{person}{Efstathia Bura}.}
  \bibinfo{year}{2022}\natexlab{}.
\newblock \showarticletitle{Moment-Based Invariants for Probabilistic Loops
  with Non-polynomial Assignments}. In \bibinfo{booktitle}{\emph{Proc. of
  {QEST} 2022: the 19th Intern. Conference on Quantitative Evaluation of
  Systems}} \emph{(\bibinfo{series}{LNCS}, Vol.~\bibinfo{volume}{13479})}.
  \bibinfo{publisher}{Springer}, \bibinfo{pages}{3--25}.
\newblock
\urldef\tempurl%
\url{https://doi.org/10.1007/978-3-031-16336-4\_1}
\showDOI{\tempurl}


\bibitem[Meecham and Jeng(1968)]%
        {MeechamJeng1968}
\bibfield{author}{\bibinfo{person}{William~C. Meecham} {and}
  \bibinfo{person}{Dah-Teng Jeng}.} \bibinfo{year}{1968}\natexlab{}.
\newblock \showarticletitle{Use of the {W}iener—{H}ermite expansion for
  nearly normal turbulence}.
\newblock \bibinfo{journal}{\emph{Journal of Fluid Mechanics}}
  \bibinfo{volume}{32}, \bibinfo{number}{2} (\bibinfo{year}{1968}),
  \bibinfo{pages}{225–249}.
\newblock
\urldef\tempurl%
\url{https://doi.org/10.1017/S0022112068000698}
\showDOI{\tempurl}


\bibitem[Moosbrugger et~al\mbox{.}(2021a)]%
        {MoosbruggerBKK21}
\bibfield{author}{\bibinfo{person}{Marcel Moosbrugger}, \bibinfo{person}{Ezio
  Bartocci}, \bibinfo{person}{Joost{-}Pieter Katoen}, {and}
  \bibinfo{person}{Laura Kov{\'{a}}cs}.} \bibinfo{year}{2021}\natexlab{a}.
\newblock \showarticletitle{Automated Termination Analysis of Polynomial
  Probabilistic Programs}. In \bibinfo{booktitle}{\emph{Proc. of {ESOP} 2021:
  the 30th European Symposium on Programming Languages and Systems}}
  \emph{(\bibinfo{series}{LNCS}, Vol.~\bibinfo{volume}{12648})}.
  \bibinfo{publisher}{Springer}, \bibinfo{pages}{491--518}.
\newblock
\urldef\tempurl%
\url{https://doi.org/10.1007/978-3-030-72019-3\_18}
\showDOI{\tempurl}


\bibitem[Moosbrugger et~al\mbox{.}(2021b)]%
        {MoosbruggerBKK21b}
\bibfield{author}{\bibinfo{person}{Marcel Moosbrugger}, \bibinfo{person}{Ezio
  Bartocci}, \bibinfo{person}{Joost{-}Pieter Katoen}, {and}
  \bibinfo{person}{Laura Kov{\'{a}}cs}.} \bibinfo{year}{2021}\natexlab{b}.
\newblock \showarticletitle{The Probabilistic Termination Tool Amber}. In
  \bibinfo{booktitle}{\emph{Proc. of {FM} 2021: the 24th International
  Symposium on Formal Methods}} \emph{(\bibinfo{series}{LNCS},
  Vol.~\bibinfo{volume}{13047})}. \bibinfo{publisher}{Springer},
  \bibinfo{pages}{667--675}.
\newblock
\urldef\tempurl%
\url{https://doi.org/10.1007/978-3-030-90870-6\_36}
\showDOI{\tempurl}


\bibitem[Moosbrugger et~al\mbox{.}(2022)]%
        {Moosbruggeretal2022}
\bibfield{author}{\bibinfo{person}{Marcel Moosbrugger},
  \bibinfo{person}{Miroslav Stankovic}, \bibinfo{person}{Ezio Bartocci}, {and}
  \bibinfo{person}{Laura Kov{\'{a}}cs}.} \bibinfo{year}{2022}\natexlab{}.
\newblock \showarticletitle{This Is The Moment for Probabilistic Loops}.
\newblock \bibinfo{journal}{\emph{Proc. {ACM} Program. Lang.}}
  \bibinfo{volume}{6}, \bibinfo{number}{{OOPSLA2}} (\bibinfo{year}{2022}),
  \bibinfo{pages}{1497–--1525}.
\newblock
\urldef\tempurl%
\url{https://doi.org/10.1145/3563341}
\showDOI{\tempurl}


\bibitem[Mühlpfordt et~al\mbox{.}(2018)]%
        {Muehlpfordtetal2018}
\bibfield{author}{\bibinfo{person}{Tillmann Mühlpfordt}, \bibinfo{person}{Rolf
  Findeisen}, \bibinfo{person}{Veit Hagenmeyer}, {and} \bibinfo{person}{Timm
  Faulwasser}.} \bibinfo{year}{2018}\natexlab{}.
\newblock \showarticletitle{Comments on Truncation Errors for Polynomial Chaos
  Expansions}.
\newblock \bibinfo{journal}{\emph{IEEE Control Systems Letters}}
  \bibinfo{volume}{2}, \bibinfo{number}{1} (\bibinfo{year}{2018}),
  \bibinfo{pages}{169--174}.
\newblock
\urldef\tempurl%
\url{https://doi.org/10.1109/LCSYS.2017.2778138}
\showDOI{\tempurl}


\bibitem[Neher et~al\mbox{.}(2007)]%
        {Neher2007}
\bibfield{author}{\bibinfo{person}{Markus Neher}, \bibinfo{person}{Kenneth~R.
  Jackson}, {and} \bibinfo{person}{Nedialko~S. Nedialkov}.}
  \bibinfo{year}{2007}\natexlab{}.
\newblock \showarticletitle{On Taylor Model Based Integration of ODEs}.
\newblock \bibinfo{journal}{\emph{SIAM J. Numer. Anal.}} \bibinfo{volume}{45},
  \bibinfo{number}{1} (\bibinfo{year}{2007}), \bibinfo{pages}{236--262}.
\newblock
\urldef\tempurl%
\url{https://doi.org/10.1137/050638448}
\showDOI{\tempurl}


\bibitem[Pairet et~al\mbox{.}(2022)]%
        {Pairetetal_2020}
\bibfield{author}{\bibinfo{person}{{\`{E}}ric Pairet},
  \bibinfo{person}{Juan~David Hern{\'{a}}ndez}, \bibinfo{person}{Marc
  Carreras}, \bibinfo{person}{Yvan~R. Petillot}, {and} \bibinfo{person}{Morteza
  Lahijanian}.} \bibinfo{year}{2022}\natexlab{}.
\newblock \showarticletitle{Online Mapping and Motion Planning Under
  Uncertainty for Safe Navigation in Unknown Environments}.
\newblock \bibinfo{journal}{\emph{{IEEE} Trans Autom. Sci. Eng.}}
  \bibinfo{volume}{19}, \bibinfo{number}{4} (\bibinfo{year}{2022}),
  \bibinfo{pages}{3356--3378}.
\newblock
\urldef\tempurl%
\url{https://doi.org/10.1109/TASE.2021.3118737}
\showDOI{\tempurl}


\bibitem[Revol et~al\mbox{.}(2005)]%
        {Revol2005}
\bibfield{author}{\bibinfo{person}{Nathalie Revol}, \bibinfo{person}{Kyoko
  Makino}, {and} \bibinfo{person}{Martin Berz}.}
  \bibinfo{year}{2005}\natexlab{}.
\newblock \showarticletitle{Taylor models and floating-point arithmetic: proof
  that arithmetic operations are validated in COSY}.
\newblock \bibinfo{journal}{\emph{The Journal of Logic and Algebraic
  Programming}} \bibinfo{volume}{64}, \bibinfo{number}{1}
  (\bibinfo{year}{2005}), \bibinfo{pages}{135--154}.
\newblock
\showISSN{1567-8326}
\urldef\tempurl%
\url{https://doi.org/10.1016/j.jlap.2004.07.008}
\showDOI{\tempurl}


\bibitem[Sankaranarayanan(2020)]%
        {Sankaranarayanan2020}
\bibfield{author}{\bibinfo{person}{Sriram Sankaranarayanan}.}
  \bibinfo{year}{2020}\natexlab{}.
\newblock \bibinfo{booktitle}{\emph{Quantitative analysis of programs with
  probabilities and concentration of measure inequalities}}.
\newblock \bibinfo{publisher}{Cambridge University Press},
  \bibinfo{pages}{259–294}.
\newblock
\urldef\tempurl%
\url{https://doi.org/10.1017/9781108770750.009}
\showDOI{\tempurl}


\bibitem[Sankaranarayanan et~al\mbox{.}(2020)]%
        {Srirametal2020}
\bibfield{author}{\bibinfo{person}{Sriram Sankaranarayanan},
  \bibinfo{person}{Yi Chou}, \bibinfo{person}{Eric Goubault}, {and}
  \bibinfo{person}{Sylvie Putot}.} \bibinfo{year}{2020}\natexlab{}.
\newblock \showarticletitle{Reasoning about Uncertainties in Discrete-Time
  Dynamical Systems using Polynomial Forms.}. In
  \bibinfo{booktitle}{\emph{Advances in Neural Information Processing
  Systems}}, Vol.~\bibinfo{volume}{33}. \bibinfo{publisher}{Curran Associates,
  Inc.}, \bibinfo{pages}{17502--17513}.
\newblock
\urldef\tempurl%
\url{https://proceedings.neurips.cc/paper/2020/file/ca886eb9edb61a42256192745c72cd79-Paper.pdf}
\showURL{%
\tempurl}


\bibitem[Son and Du(2020)]%
        {SonDu2020}
\bibfield{author}{\bibinfo{person}{Jeongeun Son} {and}
  \bibinfo{person}{Yuncheng Du}.} \bibinfo{year}{2020}\natexlab{}.
\newblock \showarticletitle{Probabilistic surrogate models for uncertainty
  analysis: Dimension reduction-based polynomial chaos expansion}.
\newblock \bibinfo{journal}{\emph{Internat. J. Numer. Methods Engrg.}}
  \bibinfo{volume}{121}, \bibinfo{number}{6} (\bibinfo{year}{2020}),
  \bibinfo{pages}{1198--1217}.
\newblock
\urldef\tempurl%
\url{https://doi.org/10.1002/nme.6262}
\showDOI{\tempurl}


\bibitem[Stankovi{\'{c}}(1996)]%
        {Stankovic1996}
\bibfield{author}{\bibinfo{person}{B. Stankovi{\'{c}}}.}
  \bibinfo{year}{1996}\natexlab{}.
\newblock \showarticletitle{Taylor Expansion for Generalized Functions}.
\newblock \bibinfo{journal}{\emph{Journal of Mathematical Analysis and
  Application}}  \bibinfo{volume}{203} (\bibinfo{year}{1996}),
  \bibinfo{pages}{31--37}.
\newblock
\urldef\tempurl%
\url{https://doi.org/10.1006/jmaa.1996.0365}
\showDOI{\tempurl}


\bibitem[Stankovic et~al\mbox{.}(2022)]%
        {StankovicBK22}
\bibfield{author}{\bibinfo{person}{Miroslav Stankovic}, \bibinfo{person}{Ezio
  Bartocci}, {and} \bibinfo{person}{Laura Kov{\'{a}}cs}.}
  \bibinfo{year}{2022}\natexlab{}.
\newblock \showarticletitle{Moment-based analysis of Bayesian network
  properties}.
\newblock \bibinfo{journal}{\emph{Theor. Comput. Sci.}}  \bibinfo{volume}{903}
  (\bibinfo{year}{2022}), \bibinfo{pages}{113--133}.
\newblock
\urldef\tempurl%
\url{https://doi.org/10.1016/j.tcs.2021.12.021}
\showDOI{\tempurl}


\bibitem[Steinhardt and Tedrake(2012)]%
        {SteinhardtT12}
\bibfield{author}{\bibinfo{person}{Jacob Steinhardt} {and}
  \bibinfo{person}{Russ Tedrake}.} \bibinfo{year}{2012}\natexlab{}.
\newblock \showarticletitle{Finite-time regional verification of stochastic
  non-linear systems}.
\newblock \bibinfo{journal}{\emph{Int. J. Robotics Res.}} \bibinfo{volume}{31},
  \bibinfo{number}{7} (\bibinfo{year}{2012}), \bibinfo{pages}{901--923}.
\newblock
\urldef\tempurl%
\url{https://doi.org/10.1177/0278364912444146}
\showDOI{\tempurl}


\bibitem[Taylor(1993)]%
        {Taylor1993}
\bibfield{author}{\bibinfo{person}{John~B. Taylor}.}
  \bibinfo{year}{1993}\natexlab{}.
\newblock \showarticletitle{{Discretion versus policy rules in practice}}.
\newblock \bibinfo{journal}{\emph{Carnegie-Rochester Conference Series on
  Public Policy}} \bibinfo{volume}{39}, \bibinfo{number}{1}
  (\bibinfo{date}{December} \bibinfo{year}{1993}), \bibinfo{pages}{195--214}.
\newblock
\urldef\tempurl%
\url{https://ideas.repec.org/a/eee/crcspp/v39y1993ip195-214.html}
\showURL{%
\tempurl}


\bibitem[Triebel(2001)]%
        {Triebel2001}
\bibfield{author}{\bibinfo{person}{Hans Triebel}.}
  \bibinfo{year}{2001}\natexlab{}.
\newblock \showarticletitle{Taylor expansions of distributions}.
\newblock In \bibinfo{booktitle}{\emph{The Structure of Functions}}.
  \bibinfo{publisher}{Birkh\"{a}user Basel}.
\newblock
\urldef\tempurl%
\url{https://doi.org/10.1007/978-3-0348-8257-6\_8}
\showDOI{\tempurl}


\bibitem[van~den Berg et~al\mbox{.}(2011)]%
        {vandenBergetal_2011}
\bibfield{author}{\bibinfo{person}{Jur van~den Berg}, \bibinfo{person}{Pieter
  Abbeel}, {and} \bibinfo{person}{Kenneth~Y. Goldberg}.}
  \bibinfo{year}{2011}\natexlab{}.
\newblock \showarticletitle{{LQG-MP:} Optimized path planning for robots with
  motion uncertainty and imperfect state information}.
\newblock \bibinfo{journal}{\emph{Int. J. Robotics Res.}}  \bibinfo{volume}{30}
  (\bibinfo{year}{2011}), \bibinfo{pages}{895--913}.
\newblock
Issue 7.
\urldef\tempurl%
\url{https://doi.org/10.1177/0278364911406562}
\showDOI{\tempurl}


\bibitem[Wan and Karniadakis(2005)]%
        {WanKarniadakis2005}
\bibfield{author}{\bibinfo{person}{Xiaoliang Wan} {and}
  \bibinfo{person}{George~E. Karniadakis}.} \bibinfo{year}{2005}\natexlab{}.
\newblock \showarticletitle{An adaptive multi-element generalized polynomial
  chaos method for stochastic differential equations}.
\newblock \bibinfo{journal}{\emph{J. Comput. Phys.}}  \bibinfo{volume}{209}
  (\bibinfo{year}{2005}), \bibinfo{pages}{617--642}.
\newblock
\urldef\tempurl%
\url{https://doi.org/10.1016/j.jcp.2005.03.023}
\showDOI{\tempurl}


\bibitem[Xiu(2010)]%
        {Xiu2010}
\bibfield{author}{\bibinfo{person}{Dongbin Xiu}.}
  \bibinfo{year}{2010}\natexlab{}.
\newblock \bibinfo{booktitle}{\emph{Numerical Methods for Stochastic
  Computations: A Spectral Method Approach}}.
\newblock \bibinfo{publisher}{Princeton University Press}.
\newblock
\showISBNx{9780691142128}
\urldef\tempurl%
\url{http://www.jstor.org/stable/j.ctv7h0skv}
\showURL{%
\tempurl}


\bibitem[Xiu and Karniadakis(2002)]%
        {XiuKarniadakis2002a}
\bibfield{author}{\bibinfo{person}{Dongbin Xiu} {and}
  \bibinfo{person}{George~Em Karniadakis}.} \bibinfo{year}{2002}\natexlab{}.
\newblock \showarticletitle{The {W}iener-{A}skey Polynomial Chaos for
  Stochastic Differential Equations}.
\newblock \bibinfo{journal}{\emph{SIAM J. Sci. Comput.}} \bibinfo{volume}{24},
  \bibinfo{number}{2} (\bibinfo{date}{Feb.} \bibinfo{year}{2002}),
  \bibinfo{pages}{619–644}.
\newblock
\showISSN{1064-8275}
\urldef\tempurl%
\url{https://doi.org/10.1137/S1064827501387826}
\showDOI{\tempurl}


\end{thebibliography}

\appendix

\section{Proof of Theorem \ref{err-bound}}\label{app:A}

\begin{proof}[Thm. \ref{err-bound}]
Since $f(z) = 0$ $\forall z\notin \left[a, b\right]$,
\begin{gather}
  \left \| g(z) - \sum\limits_{i=0}^{T}c_{i}p_{i}(z)\right \|_{f}^{2} = 
\int\limits_{a}^{b}\left(g(z) - \sum\limits_{i=0}^{T}c_{i}p_{i}(z)\right)^{2}f(z)dz 
=\int\limits_{-\infty}^{\infty}\left(g(z) - \sum\limits_{i=0}^{T}c_{i}p_{i}(z)\right)^{2}f(z)dz\notag \\
= \int\limits_{-\infty}^{a}\left(g(z) - \sum\limits_{i=0}^{T}c_{i}p_{i}(z)\right)^{2}f(z)dz 
+ \int\limits_{b}^{\infty}\left(g(z) - \sum\limits_{i=0}^{T}c_{i}p_{i}(z)\right)^{2}f(z)dz \notag \\
+ \int\limits_{a}^{b}\left(g(z) - \sum\limits_{i=0}^{T}c_{i}p_{i}(z)\right)^{2}f(z)dz  \notag \\
\leq \int\limits_{-\infty}^{a}\left(g(z) - \sum\limits_{i=0}^{T}c_{i}p_{i}(z)\right)^{2}\phi(z)dz 
 + \int\limits_{b}^{\infty}\left(g(z) - \sum\limits_{i=0}^{T}c_{i}p_{i}(z)\right)^{2}\phi(z)dz \notag \\
+ \int\limits_{a}^{b}\left(g(z) - \sum\limits_{i=0}^{T}c_{i}p_{i}(z)\right)^{2}\phi(z)dz  
+ \int_{a}^{b}\left(g(z) - \sum\limits_{i=0}^{T}c_{i}p_{i}(z)\right)^{2}(f(z)-\phi(z))dz \notag \\
= A + B + C + D \label{star}
\end{gather}
Since $ f(z) -\phi(z)  \leq  \phi(z) + f(z)$, $D$ satisfies 
\begin{gather*}
\int\limits_{a}^{b}\left(g(z) - \sum\limits_{i=0}^{T}c_{i}p_{i}(z)\right)^{2}(f(z)-\phi(z))dz  \leq \int\limits_{a}^{b}\left(g(z) - \sum\limits_{i=0}^{T}c_{i}p_{i}(z)\right)^{2}dz\int\limits_{a}^{b}(\phi(z) + f(z))dz  \\
= (1 + \Phi(b) - \Phi(a)) \times  \int\limits_{a}^{b}\left(g(z) - \sum\limits_{i=0}^{T}c_{i}p_{i}(z)\right)^{2}dz,
\end{gather*}
with $(1 + \Phi(b) - \Phi(a)) < 2.$
Now, $\forall z \in \left[a, b\right],$ $1 \leq \phi(z)/\min((\phi(a), \phi(b)))$, and hence 
\begin{gather}
   \int\limits_{a}^{b}\left(g(z) - \sum\limits_{i=0}^{T}c_{i}p_{i}(z)\right)^{2}dz 
   \leq  \min{(\phi(a), \phi(b))}^{-1}\int\limits_{a}^{b}\left(g(z) - \sum\limits_{i=0}^{T}c_{i}p_{i}(z)\right)^{2}\phi(x)dz \notag\\
\leq \min{(\phi(a), \phi(b))}^{-1} C. \label{2nd} 
\end{gather} 
By \eqref{2nd} and \eqref{truef}, \eqref{star} satisfies
\begin{gather}
 A+B+C+D \leq \left(\frac{2}{\min{(\phi(a), \phi(b))}} + 1 \right)  
\Bigg\{ \int\limits_{-\infty}^{a} \left(g(z) - \sum\limits_{i=0}^{T}c_{i}p_{i}(z)\right)^{2}\phi(z)dz \notag \\ + 
\int\limits_{b}^{\infty}\left(g(z) - \sum\limits_{i=0}^{T}c_{i}p_{i}(z)\right)^{2}\phi(z)dz 
 +  \int\limits_{a}^{b}\left(g(z) - \sum\limits_{i=0}^{T}c_{i}p_{i}(z)\right)^{2}\phi(z)dz  \Bigg\}  \notag \\
= \left(\frac{2}{\min{(\phi(a), \phi(b))}} + 1 \right) 
  \int\limits_{-\infty}^{\infty}\left(g(z) - \sum\limits_{i=0}^{T}c_{i}p_{i}(z)\right)^{2}\phi(z)dz \notag\\
   = \left(\frac{2}{\min{(\phi(a), \phi(b))}} + 1 \right)\sum\limits_{i=T+1}^{\infty}c_{i}^{2} \leq \left( \frac{2}{\min{(\phi(a), \phi(b))}} + 1 \right) \var_{\phi}\left[g(Z)\right]  \notag 
\end{gather} 
since $\var_\phi(g(Z))=\sum_{i=1}^\infty c_i^2$.
In consequence, the error \eqref{approxer} 
can be upper bounded by \eqref{bound}.
\end{proof}

\section{PCE of exponential and trigonometric functions}\label{app:C}

Table \ref{tab:LiOFunc} lists examples of functions of up to three random arguments approximated by PCE's of different degrees and, correspondingly, number of coefficients. We use $TruncNormal \left(\mu, \sigma^{2}, \left[a, b\right]\right)$ to denote the truncated normal distribution with expectation $\mu$ and  standard deviation $\sigma$ on the (finite or infinite) interval $\left[a, b\right]$, and $TruncGamma\left(\theta, k, \left[a, b\right]\right)$ for the  truncated gamma distribution on the (finite or infinite) interval $\left[a, b\right]$, $a,b>0$, with shape parameter $k$ and scale parameter $\theta$.
The  approximation error in \eqref{error_std} is reported in the last column. The results confirm \eqref{convergence} in practice: the error decreases as the degree or, equivalently, the number of  components in the approximation of the polynomial increases.  

\begin{table}[h!]
\centering
\setlength\belowcaptionskip{8pt}
\resizebox{\textwidth}{!}{
\begin{tabular}{@{}llcc@{}} 
 \toprule
 Function & Random Variables & Degree / \#coefficients & Error\\ \midrule
 
 \begin{tabular}{@{}l@{}}$f(x_{1}, x_{2}) = \xi e^{-x_{1}} + (\xi - \frac{\xi^{2}}{2}) e^{x_{2}-x_{1}}$ \\ $\xi = 0.3$ \end{tabular} & 
  \begin{tabular}{@{}l@{}}$x_{1} \sim Normal(0, 1),$ \\ $x_{2} \sim Normal(2, 0.01)$ \end{tabular} & 
  \begin{tabular}{@{}c@{}} 1 / 4 \\ 2 / 9 \\ 3 / 16 \\ 4 / 25 \\ 5 / 36 \end{tabular} & 
  \begin{tabular}{@{}c@{}} 3.076846 \\ 1.696078 \\ 0.825399 \\ 0.363869 \\ 0.270419 \end{tabular} \\ 
  \hline

 $f(x_{1}, x_{2}) = 0.3 e^{x_{1} - x_{2}} + 0.6 e^{ - x_{2}}$ & 
  \begin{tabular}{@{}l@{}}$x_{1} \sim TruncNormal(4, 1, \left[3, 5\right]),$ \\ $x_{2} \sim TruncNormal(2, 0.01, \left[0, 4\right])$ \end{tabular} & 
  \begin{tabular}{@{}c@{}} 1 / 4 \\ 2 / 9 \\ 3 / 16 \\ 4 / 25 \\ 5 / 36 \end{tabular} & 
  \begin{tabular}{@{}c@{}} 0.343870 \\ 0.057076 \\ 0.007112 \\ 0.000709 \\ 0.000059 \end{tabular} \\ 
  \hline

 $f(x_{1}, x_{2}) = e^{x_{1}x_{2}}$ & 
 \begin{tabular}{@{}l@{}}$x_{1} \sim TruncNormal(4, 1, \left[3, 5\right])$ \\ $x_{2} \sim TruncGamma(3, 1, \left[0.5, 1\right])$ \end{tabular} & 
 \begin{tabular}{@{}c@{}} 1 / 4 \\ 2 / 9 \\ 3 / 16 \\ 4 / 25 \\ 5 / 36 \end{tabular} & 
 \begin{tabular}{@{}c@{}} 5.745048 \\ 1.035060 \\ 0.142816 \\ 0.016118 \\ 0.001543 \end{tabular}\\ 
 \hline

 \begin{tabular}{@{}l@{}}$f(x_{1}, x_{2}, x_{3}) = 0.3 e^{x_{1} - x_{2}} + $ \\ \hspace{2.25cm} $ 0.6e^{x_{2} - x_{3}} +  0.1 e^{x_{3} - x_{1}}$  \end{tabular} &
 \begin{tabular}{@{}l@{}}$x_{1} \sim TruncNormal(4, 1, \left[3, 5\right])$ \\ $x_{2} \sim TruncGamma(3, 1, \left[0.5, 1\right])$ \\ $x_{3} \sim U\left[4, 8\right]$\end{tabular} & 
 \begin{tabular}{@{}c@{}} 1 / 8 \\ 2 / 27 \\ 3 / 64 \end{tabular} & 
 \begin{tabular}{@{}c@{}} 1.637981 \\ 0.303096 \\ 0.066869 \end{tabular}\\
 \hline

 \begin{tabular}{@{}l@{}}$f(x_{1}) = \psi cos(x_{1}) + (1 - \psi)sin(x_{1})$ \\ $\psi = 0.3$  \end{tabular} & $x_{1} \sim Normal(0, 1)$ & 
 \begin{tabular}{@{}c@{}} 1 / 2 \\ 2 / 3 \\ 3 / 4 \\ 4 / 5 \\ 5 / 6 \end{tabular} &  
  \begin{tabular}{@{}c@{}} 0.222627 \\ 0.181681 \\ 0.054450 \\ 0.039815 \\ 0.009115 \end{tabular}\\ 
\bottomrule
\end{tabular}
}
\caption{
Approximations of $5$ non-linear functions using PCE.
}
\label{tab:LiOFunc}
\end{table}

\section{Trigonometric Identities}\label{app:D}

We use the  following properties of $\sin$, $\cos$ and $\exp$ functions. 
\begin{align*}
    \sin(\alpha \pm \beta) = \sin(\alpha)\cos(\beta) \pm \cos(\alpha)\sin(\beta)\\
    \cos(\alpha \pm \beta) = \cos(\alpha)\cos(\beta) \mp \sin(\alpha)\sin(\beta)\\
    \sin(n\alpha) = \sum_{\substack{r = 0,\\ 2r+1 \leq n}}(-1)^{r}\binom{n}{2r + 1}\cos^{n - 2r - 1}(\alpha)\sin^{2r + 1}(\alpha), n \in \nat \\
    \cos(n\alpha) = \sum_{\substack{r = 0,\\ 2r \leq n}}(-1)^{r}\binom{n}{2r}\cos^{n - 2r}(\alpha)\sin^{2r}(\alpha), n \in \nat \\
\end{align*}

\end{document}